\documentclass{article}
\usepackage{graphicx} 

\usepackage[margin=1in]{geometry}
\usepackage{pdflscape}
\usepackage{titlesec}
\usepackage[LGR,T1]{fontenc}
\usepackage{hyperref}
\hypersetup{
    colorlinks,
    linkcolor={red},
    citecolor={orange},
    urlcolor={blue!80!black}
}

\makeatletter
\def\bbordermatrix#1{\begingroup \m@th
  \@tempdima 4.75\p@
  \setbox\z@\vbox{%
    \def\cr{\crcr\noalign{\kern2\p@\global\let\cr\endline}}%
    \ialign{$##$\hfil\kern2\p@\kern\@tempdima&\thinspace\hfil$##$\hfil
      &&\quad\hfil$##$\hfil\crcr
      \omit\strut\hfil\crcr\noalign{\kern-\baselineskip}%
      #1\crcr\omit\strut\cr}}%
  \setbox\tw@\vbox{\unvcopy\z@\global\setbox\@ne\lastbox}%
  \setbox\tw@\hbox{\unhbox\@ne\unskip\global\setbox\@ne\lastbox}%
  \setbox\tw@\hbox{$\kern\wd\@ne\kern-\@tempdima\left[\kern-\wd\@ne
    \global\setbox\@ne\vbox{\box\@ne\kern2\p@}%
    \vcenter{\kern-\ht\@ne\unvbox\z@\kern-\baselineskip}\,\right]$}%
  \null\;\vbox{\kern\ht\@ne\box\tw@}\endgroup}
\makeatother

\usepackage{graphicx}
\usepackage{subcaption}
\usepackage{tikz}
\usepackage{pgfplots}
\pgfplotsset{compat=1.18} 
\usetikzlibrary{angles,
                graphs,
                graphs.standard, 
                quantikz2,
                decorations.pathmorphing,
                patterns, 
                patterns.meta,
                shadows}
\tikzset{snake it/.style={decorate, decoration=snake}}

\usepackage{physics}
\usepackage{amsmath,amsfonts, amsthm}
\usepackage{pifont}
\usepackage{dsfont}

\usepackage{csvsimple}
\usepackage{listings}
\usepackage{minted}

\usepackage{paralist}

\definecolor{tuyen}{rgb}{0.3, 0.4, 0.8}

\definecolor{amira}{rgb}{0, 0, 1}

\usepackage{booktabs}
\usepackage{makecell}

\usepackage{mdframed}
    \newmdtheoremenv{defn}{Definition}
    \newmdtheoremenv{thrm}{Theorem}
    \newmdtheoremenv{thrm*}{Theorem}
    \newmdtheoremenv{hlo}{Overview}
    \newmdtheoremenv{fact}{Fact}
    \newmdtheoremenv{remark}{Remark}
    \newmdtheoremenv{lmma}{Lemma}
    \newmdtheoremenv{prop}{Proposition}
    \newmdtheoremenv{stm}{Statement}
    \newmdtheoremenv{corollary}{Corollary}
    \newmdtheoremenv{claim}{Claim}
    \newmdtheoremenv{hypothesis}{Hyposthesis}
\newenvironment{proofsketch}{%
  \proof}{\endproof}
\usepackage{algorithm}
\usepackage{algpseudocode}   

\DeclareMathAlphabet{\mathgtt}{LGR}{cmtt}{m}{n}

\usepackage[toc,page]{appendix}

\usepackage{biblatex}
\usepackage{preamble}

\newcommand{\QSQ}{\mathsf{QSQ}}
\newcommand{\QQC}{\mathsf{QQC}}

\newcommand{\Qstat}{\mathsf{Qstat}}
\newcommand{\Stat}{\mathsf{Stat}}
\newcommand{\QSD}{\mathsf{QSD}}

\newcommand{\QAC}{\mathsf{QAC}}
\newcommand{\SQ}{\mathsf{SQ}}

\newcommand{\TRD}{d_{\text{tr}}}

\newcommand{\PR}{\mathbb{P}}

\usepackage[affil-it]{authblk}

\title{\Large{\textbf{On Quantum Learning Advantage Under Symmetries}}}
\author[1]{\normalsize Tuyen Nguyen\thanks{ \href{mailto:tuyen.q.nguyen@student.uts.edu.au}{tuyen.q.nguyen@student.uts.edu.au}}}
\author[1]{\normalsize  M\'aria Kieferov\'a}
\author[2]{\normalsize Amira Abbas}
\affil[1]{\small University of Technology Sydney, NSW 2007, Australia
}
\affil[2]{\small Google Quantum AI, Venice, California, 90291, USA.
}
\date{}
\begin{document}
\maketitle

\begin{abstract}%
   Symmetry underlies many of the most effective classical and quantum learning algorithms, yet whether quantum learners can gain a fundamental advantage under symmetry-imposed structures remains an open question. Based on evidence that classical statistical query ($\SQ$) frameworks have revealed exponential query complexity in learning symmetric function classes, we ask: can quantum learning algorithms exploit the problem symmetry better? In this work, we investigate the potential benefits of symmetry within the quantum statistical query ($\QSQ$) model, which is a natural quantum analog of classical $\SQ$. Our results uncover three distinct phenomena: (i) we obtain an exponential separation between $\QSQ$ and $\SQ$ on a permutation-invariant function class; (ii) we establish query complexity lower bounds for $\QSQ$ learning that match, up to constant factors, the corresponding classical $\SQ$ lower bounds for most commonly studied symmetries; however, the potential advantages may occur under highly skewed orbit distributions; and (iii) we further identify a tolerance-based separation exists, where quantum learners succeed at noise levels that render classical $\SQ$ algorithms ineffective. Together, these findings provide insight into when symmetry can enable quantum advantage in learning.
\end{abstract}

\section{Introduction}
Symmetry has long played a central role in the design and analysis of machine learning models. In classical learning, widely used architectures such as convolutional and graph neural networks encode translation and permutation invariance as geometric inductive biases, yielding significant improvements in efficiency and accuracy~\cite{lecun2002gradient, wu2022graph, bronstein2021geometric, batzner20223}. Beyond empirical performance, a growing body of theoretical work has demonstrated that encoding symmetry as an architectural inductive bias can reduce sample complexity and tighten generalization bounds, thereby providing a principled explanation for the observed success of symmetry-aware models~\cite{cohen2016group, long2019generalization, mei2021learning, tahmasebi2023exact}. The importance of symmetry extends just as strongly to the quantum setting. Quantum models are inherently constrained by group symmetries, from the representation theory of quantum states to the design of equivariant quantum neural networks~\cite{schatzki2024theoretical,West_2024,pesah2021absence}. Symmetry-inspired constructions of quantum circuits have been proposed as remedies to central obstacles in quantum neural networks, such as barren plateaus~\cite{mcclean2018barren} and excessive local minima~\cite{anschuetz2022quantum}. These proposals suggest that encoding symmetry into circuit architectures can stabilize optimization landscapes, reduce the number of parameters, and capture structural regularities in data, closely mirroring classical geometric deep learning~\cite{bronstein2021geometric}.

Despite these promising directions, it remains poorly understood whether symmetry can deliver a quantum advantage over its classical counterpart. This motivates a fundamental question:
\begin{center}
\emph{Can symmetry enable quantum advantage in learning tasks?}
\end{center}
This question has been partially addressed in recent work~\cite{Cerezo_2025}, which studies the classical simulability of quantum neural networks designed to avoid barren plateaus. That work shows that enforcing symmetry, such as $\mathbb{S}_n$-equivariance, can provably improve trainability, but may simultaneously restrict expressivity to the extent that the resulting models become efficiently classically simulable. It remains unclear, however, whether this phenomenon extends beyond equivariant quantum neural networks to more general quantum learning algorithms.

In contrast, existing results indicate that learning symmetric function classes can be exponentially hard in the statistical query ($\SQ$) model~\cite{kiani2024hardness}. The $\SQ$ framework captures a broad family of classical learning algorithms, including gradient-based methods, convex optimization, and Markov-chain Monte Carlo~\cite{reyzin2020statistical}. In particular, it demonstrates that while symmetry can dramatically reduce the amount of data required to represent a target function, the computational cost of learning such functions may remain prohibitively large.~\cite{kiani2024hardness} proved that imposition of symmetry on a function class only reduces the computational complexity of learning by a factor proportional to the size of the group.


In this work, we ask: can \textit{quantum statistical query} ($\QSQ$) algorithms exploit the problem symmetry better? The $\QSQ$ model, introduced by \cite{arunachalam2020quantumstatisticalquerylearning}, extends the classical statistical query ($\SQ$) framework of \cite{kearns1998efficient} with the goal of identifying potential quantum advantages in learning classical function classes. Our aim is therefore to have a direct, task-matched comparison between the quantum and classical settings for learning symmetric functions, as studied in~\cite{kiani2024hardness}. We compare the two models through their \emph{query complexity}, defined as the number of queries to the respective statistical oracles required to learn a target function class to a prescribed accuracy. Each query is answered up to a tolerance parameter $\tau$, which controls the precision of the statistical information returned by the oracle. Here, $\Stat$ denotes the classical statistical oracle and $\Qstat$ its quantum analogue; formal definitions are given in a later section. This query-complexity perspective offers a clean, information-theoretic measure of learning power, which we will employ to rigorously compare the capabilities of classical and quantum models.

\subsection{Our results}
 We provide a unified characterization of how group-invariant structure governs the learnability of both classical and quantum models within the statistical query framework. Our results reveal a sharp dichotomy: while symmetry can enable dramatic quantum advantages in special cases, it also imposes fundamental limitations that collapse quantum learning to classical hardness in broad regimes. Our main contributions are summarized below.

\medskip
\noindent
\textbf{Permutation-invariant separation.}
We begin by identifying a setting that admits an exponential separation between quantum and classical statistical query models.

\begin{thrm}[Permutation-invariant separation (Informal; see Thms.~\ref{thrm: classical_hardness_permutation} and~\ref{thrm: quantum_parity})]\phantom{a}
There exists a permutation-invariant function class such that, for any tolerance parameter $\tau \in (0,1/2)$, the class can be learned exactly by a $\QSQ$ algorithm using $\mathcal{O}(n)$ queries to $\Qstat$, whereas any $\SQ$ algorithm requires at least $\tau^{2}\exp\!\left(n^{\Omega(1)}\right)$ queries to $\Stat$ to achieve sufficiently small classification error.
\end{thrm}

\medskip
\noindent
\textbf{A lower-bound under general symmetry.}
We then move beyond permutation invariance to arbitrary group actions. Here, we develop an orbit-based analysis that yields lower bounds for both classical and quantum statistical query models. These lower bounds reveal that the learnability of symmetric function classes is fundamentally controlled by the induced orbit structure of the group action.

\begin{thrm}[$\QSQ$ lower bounds under general symmetry (Informal; see Thms.~\ref{thrm: qsq_lb_diagonal} and~\ref{thrm: qsq_lb_general}]
Let $G$ be a symmetry group acting on $\mathcal{X}=\{0,1\}^n$ via a representation
$\rho: G \to \mathrm{GL}(\mathcal{X})$, and let
\[
\mathcal{C}_{\rho}
=
\{ f:\mathcal{X}\to\{0,1\} \mid f(\rho(g)\cdot x)=f(x)\ \text{for all } g\in G,\, x\in\mathcal{X} \}
\]
denote the corresponding class of symmetric Boolean functions.
Let $\mathcal{O}_\rho$ be the set of orbits induced by the action of $G$, and define the orbit distribution
\[
p_{\mathcal{O}_\rho}(O)
=
{|O|}/{|\mathcal{X}|}, \forall O\in\mathcal{O}_\rho.
\]
Any classical $\SQ$ algorithm that learns $\mathcal{C}_\rho$ to constant error with tolerance $\tau$ requires at least
\[
\Omega\!\left(\frac{\tau^2}{\|p_{\mathcal{O}_\rho}\|_2^2}\right)
\]
queries. Moreover, there exists $\QSQ$ algorithms that learn $\mathcal{C}_\rho$ with the same tolerance $\tau$ requires at least
\[
\Omega\!\left(\frac{\tau^2\,|\mathcal{X}|}{\max_{O_k\in\mathcal{O}_\rho}|O_k|}\right)
\]
queries.
\end{thrm}
A direct consequence of the above bounds is that, for most commonly studied symmetry groups, the quantum and classical lower bounds coincide up to constant factors. At the same time, our analysis identifies highly skewed orbit distributions as a structural regime in which such an advantage may arise.
\begin{remark}[Informal; see Rems.~\ref{rem: skewed_orbit} and~\ref{rem: common_groups}]
For most commonly studied symmetry groups, the query complexity lower bound for $\QSQ$ learning matches the classical $\SQ$ lower bound up to constant factors. However, an asymptotic quantum advantage may arise when the induced orbit distribution is highly skewed.
\end{remark}

\medskip
\noindent
\textbf{Tolerance-based separations.}
Finally, we identify a quantum advantage through the tolerant parameter. In particular, quantum learners can succeed in tolerance regimes that render classical $\SQ$ algorithms ineffective.

\begin{thrm}[Tolerance-based advantage (Informal; see Thm.~\ref{thm:tolerance_separation}]
Let $\mathcal{S}$ be a subclass of symmetric Boolean functions with disjoint supports, each having measure $\zeta$. If $\zeta < 0.4$, then for any tolerance $\tau$ satisfying
$2\zeta < \tau < \sqrt{2\zeta - \zeta^2}$, the following hold:
\begin{itemize}
    \item \textbf{Classical hardness:} No classical $\SQ$ algorithm can learn $\mathcal{S}$ with tolerance $\tau$.
    \item \textbf{Quantum learnability:} There exists a $\QSQ$ algorithm that learns $\mathcal{S}$ with tolerance $\tau$ using $O(\log |\mathcal{S}|)$ queries.
\end{itemize}
\end{thrm}

\section{Background} \label{sec: bg}
\subsection{Notation}
Throughout the paper, we use uppercase letters such as $M$, $A$, and $U$ to denote matrices, with $I$ representing the identity matrix of appropriate dimension. The symbol $\mathbf{1}$ denotes the all-ones column vector, with dimension clear from context. For two quantum states (density matrices) $\rho$ and $\sigma$, their trace distance is defined as 
$$
    \TRD(\rho,\sigma) = \tfrac{1}{2}\|\rho - \sigma\|_1 = \tfrac{1}{2}\Tr\!\left[\sqrt{(\rho - \sigma)^\dagger(\rho - \sigma)}\right],
$$ 
where $\|\cdot\|_1$ is the Schatten $1$-norm. For a matrix $M$, we denote by $\|M\|$ its operator norm, i.e., the largest singular value of $M$. For an event, we write $\mathds{1}(\cdot)$ for its indicator function, which returns $1$ if the event occurs and $0$ otherwise. We use $\text{Unif}(S)$ to denote the uniform distribution over a finite set $S$, with $\text{Unif}(\{0,1\}^n)$ representing the uniform distribution over $n$-bit strings in particular.   

\subsection{Quantum Statistical Queries Model}
\label{subsec:learning_framework_bg}

The statistical query ($\SQ$) model~\cite{kearns1998efficient} captures learning algorithms that access
a target function only through expectations of bounded query functions, rather than labeled samples.
The \emph{quantum statistical query} ($\QSQ$) model~\cite{arunachalam2020quantumstatisticalquerylearning,arunachalam2023roleentanglementstatisticslearning}
is a quantum analogue in which classical expectations are replaced by expectation values of quantum observables.

Formally, let $\mathcal{X}$ be $\{0,1\}^n$ and a concept class $\mathcal{C} := \{ f: \mathcal{X} \to \{0,1\} \}$. The goal is to identify an unknown target concept $f \in \mathcal{C}$ under a distribution $\mathcal{D}$ on $\mathcal{X}$. In the classical $\SQ$ framework, rather than examining a sequence of labeled samples $(x,f(x))$, the learner interacts with a \emph{statistical oracle} $\Stat$ is defined as follows
\begin{defn}[Statistical Oracle] Let $\phi : \mathcal{X} \times \{0,1\} \to [-1,1]$ be a query function and a tolerance parameter $\tau > 0$. For a target concept $f \in \mathcal{C}$, the statistical oracle returns
\[
\Stat(\phi,\tau) = \alpha 
\quad \text{such that} \quad
\Big| \alpha - \mathbb{E}_{x \sim \mathcal{D}}[\phi(x,f(x))] \Big| \leq \tau.
\]
With a slight abuse of notation, we denote by $\Stat(\tau)$ a query to any map $\phi$ satisfying $\phi : \mathcal{X} \times \{0,1\} \to [-1,1]$.
\end{defn}

The quantum statistical query ($\QSQ$) model generalizes this framework to the quantum setting. Instead of classical labeled examples, information about the target concept is encoded in the \emph{quantum example state}
\[
\ket{\psi_f}
=\sum_{x\in\mathcal{X}}\sqrt{\mathcal{D}(x)}\ket{x,f(x)},
\]
which is a superposition of labeled samples. This quantum example state is well-motivated in quantum computing and has many interesting complexity-theoretic applications~\cite{aharonov2003adiabatic}. Same as $\SQ$, a $\QSQ$ learner interacts with the state only through a
\emph{quantum statistical oracle}.

\begin{defn}[Quantum Statistical Oracle]\label{defn: qstat}
Let $O \in \mathbb{C}^{2^{n+1} \times 2^{n+1}}$ be any observable with operator norm $\|O\| \leq 1$, and let $\tau > 0$.  
For a target concept $f \in \mathcal{C}$ encoded in the state $\ket{\psi_{f}}$, the quantum statistical oracle returns
\[
\Qstat(O,\tau) = \alpha \quad \text{such that} \quad 
\big| \alpha - \langle \psi_{f} | O | \psi_{f} \rangle \big| \leq \tau.
\]
With a slight abuse of notation, we denote by $\Qstat(\tau)$ a query to any observable $O$ satisfying $\|O\| \leq 1$.
\end{defn}
If $O$ is diagonal in the computational basis, $\Qstat$ reduces exactly to a $\Stat$ query.
Hence, the $\QSQ$ model strictly generalizes $\SQ$.
Any potential quantum advantage must therefore arise from the use of
\emph{non-diagonal observables}, which can access global interference or entanglement-induced
correlations unavailable to classical $\SQ$ algorithms.
At the same time, the $\QSQ$ model is that the learner itself remains a {classical} randomized algorithm. It never manipulates quantum states directly, but only receives approximate numerical values $\alpha$ returned by the oracle.

In this paper, we focus on lower bounds on the query complexity of the $\QSQ$ model, which is the number of oracle calls $\Qstat(O,\tau)$ required to complete a learning task, and on how this complexity compares to that of the classical $\SQ$ model on the same problem. Our main technical tool is the \emph{quantum statistical dimension} ($\QSD$), introduced in~\cite{arunachalam2023roleentanglementstatisticslearning}.
The standard approach reduces learning to an associated decision problem
(distinguishing an unknown state $\rho\in\mathcal{C}$ from a fixed reference state
$\sigma\notin\mathcal{C}$) and lower bounds the decision complexity via $\QSD$. Hence, bounding $\QSD$ provides a direct measure of the complexity of learning in the $\QSQ$ model. Formal reduction statements are deferred to Appendix~\ref{app:qsd_reductions}.

\begin{defn}[Quantum statistical dimension ($\QSD$)] \label{defn: qsd}
Let $\mathcal{C}$ be a set of quantum states and let $\sigma\notin\mathcal{C}$.
For a distribution $\mu$ over $\mathcal{C}$, define
\[
\kappa_{\tau}\text{-frac}(\mu,\sigma)
=\max_{\|O\|\le 1}
\mathbb{P}_{\rho\sim\mu}\!\left[
\big|\Tr[O(\rho-\sigma)]\big|>\tau
\right],
\qquad
\QSD_{\tau}(\mathcal{C},\sigma)
=\sup_{\mu}\big(\kappa_{\tau}\text{-frac}(\mu,\sigma)\big)^{-1}.
\]
\end{defn}
Intuitively, $\QSD_{\tau}(\mathcal{C},\sigma)$ measures how many distinct observables are required,
at tolerance $\tau$, to reliably separate typical states in $\mathcal{C}$ from $\sigma$. A large $\QSD$ implies that no single query can distinguish more than a small fraction of the class, forcing many oracle queries. The relation of $\QSD$ and $\QSQ$ is then presented in the following lemma.
\begin{lmma}[$\QSD$ lower bounds $\QSQ$ learning complexity]
\label{lmma:qsd_bound_qsq}
Let $\varepsilon \ge \tau > 0$ and $\delta\in(0,1/2)$.  
Let $\mathcal{C}$ be a concept class of states and let $\sigma\notin\mathcal{C}$ be a reference state such that
$
\min_{\rho\in\mathcal{C}} d_{\mathrm{tr}}(\rho,\sigma) \;>\; 2(\tau+\varepsilon).
$
Then the $\QSQ$ learning complexity satisfies
\[
\QSQ^{\varepsilon,\delta}_{\tau}(\mathcal{C})
\;\ge\;
(1-2\delta)\,\QSD_{\tau}(\mathcal{C},\sigma)\;-\;1.
\]
where we denote $\QSQ^{\varepsilon,\delta}_{\tau}(\mathcal{C})$ as the number of $\Qstat(\tau)$ queries made by a $\QSQ$ algorithm that on input $\rho$ outputs $\pi \in \mathcal C$ such that $\TRD(\pi, \rho) \leq \epsilon$ with probability $\geq 1-\delta$.
\end{lmma}
This lemma shows that lower bounds on the quantum statistical dimension $\QSD$ translate directly into lower bounds on the $\QSQ$ learning complexity. However, directly computing $\QSD$ is often difficult. Following~\cite{arunachalam2023roleentanglementstatisticslearning}, we present a lower-bound technique that will be central in our later symmetry arguments.
\begin{lmma}[Variance bound~\cite{arunachalam2023roleentanglementstatisticslearning}]
\label{lmma: variance_bound}
Let $\mu$ be a distribution over $\mathcal{C}$ and let $\sigma=\mathbb{E}_{\rho\sim\mu}[\rho]\notin\mathcal{C}$.
Define
$
\mathrm{Var}_\mu(\mathcal{C})
:=\sup_{\|O\|\le 1}
\Big(
\mathbb{E}_{\rho\sim\mu}[\Tr(O\rho)^2]
-
(\mathbb{E}_{\rho\sim\mu}\Tr(O\rho))^2
\Big).
$
Then for every $\tau\in(0,1]$,
\[
\QSD_{\tau}(\mathcal{C},\sigma)
\;\ge\;
\Omega\!\left(\frac{\tau^2}{\mathrm{Var}_\mu(\mathcal{C})}\right).
\]
\end{lmma}
This variance bound analyzes the fluctuations of measurement outcomes across an ensemble of states: if, for every observable $O$, the variance of $\Tr[O(\rho-\sigma)]$ over $\rho \in \mathcal{C}$ is small, then no single query can reliably distinguish members of $\mathcal{C}$ from the reference state $\sigma$, forcing the learner to make many queries. 
\section{Main Results} \label{sec: main_results}
In this section, we study quantum statistical query ($\QSQ$) learning under symmetry across three regimes.
In Section~\ref{sec: permutation}, we show that for a permutation-invariant function class,
quantum learners achieve an exponential separation over classical $\SQ$ learners by leveraging
quantum Fourier sampling to uncover hidden parity structures that induce exponentially small
classical correlations.
In Section~\ref{sec: general}, we extend the analysis to general symmetry groups and establish
matching lower bounds for $\QSQ$ and $\SQ$ learning, showing that the orbit structure induced by
the group action constrains both models equally in most of the commonly studied symmetry groups, but the asymptotic advantage may arise under a highly skewed orbit distribution. Finally, in Section~\ref{sec: tolerance}, we examine the role of the tolerance parameter and
demonstrate that quantum learners can succeed at tolerance levels where all classical $\SQ$
learners fail.

\subsection{Learning permutation-invariant functions}\label{sec: permutation}
We begin with the special case of permutation-invariant function classes, which
provide a clean setting for illustrating separations between classical $\SQ$ and
quantum $\QSQ$ learning.
While permutation-equivariant quantum neural networks have been shown to enjoy favorable
optimization and generalization properties~\cite{schatzki2024theoretical},
recent results demonstrate that such models with bounded observables are
classically simulatable~\cite{Cerezo_2025}.
In contrast, we show that under the permutation-invariant function class, the $\QSQ$ model can yield exponential advantages over classical $\SQ$ learning.

To make this separation concrete, we consider a simple permutation-invariant
Boolean function class defined on graphs.
Let $A\in\{0,1\}^{n\times n}$ be the adjacency matrix of a directed graph sampled
uniformly at random.
Define the degree-count vector $c_A\in[n]^{n+1}$, where $[c_A]_i$ records the
number of nodes with out-degree $i-1$.
For each subset $S\subseteq[n+1]$, define
\[
g_S(A)\;=\;\sum_{i\in S}[c_A]_i \bmod 2.
\]
Let $\mathcal{C}_n=\{g_S:S\subseteq[n+1]\}$.
Each function in $\mathcal{C}_n$ is Boolean-valued and invariant under
permutations of the graph vertices.
Moreover, $\mathcal{C}_n$ can be represented exactly by two-hidden-layer graph
neural networks (GNNs)~\cite[Lemma~15]{kiani2024hardness}, a widely used
architecture for processing graph-structured data.
Consequently, our results imply that for function classes representable by
two-hidden-layer GNNs, there exist function families that are provably hard to
learn in any $\SQ$ model, yet admit efficient learning algorithms in the $\QSQ$
model.

Note that each function $g_S$ is essentially a parity over selected entries of the
degree-count vector $c_A$. As a result, the class $\mathcal{C}_n$ inherits the characteristic hardness of
parity functions in the statistical query model, leading to the following lower
bound.
\begin{thrm}[Classical Hardness, A variant of Theorem 3 in~\cite{kiani2024hardness}] Let $\tau > 0$. Any $\SQ$ algorithm that learns $\mathcal{C}_{n}$ up to classification error $\epsilon$ sufficiently small ($\epsilon < 1/4$) requires at least $ \Omega\!\left(\tau^2 \exp\!\big(n^{\Omega(1)}\big)\right)$ $\Stat(\tau)$ queries. \label{thrm: classical_hardness_permutation} \end{thrm}
\begin{proofsketch}
This $\SQ$ lower bound follows the argument of~\cite[Theorem~C.4]{chen2022learning}, which relies on the 
\emph{pairwise independence} property of the function class. 
Let $\mathcal{D}$ be a distribution over graphs $A \in \{0,1\}^{n\times n}$. 
For any $g \in \mathcal{C}_n$, define the corresponding labeled distribution
$
\mathcal{D}_{g} := \{(A,\,y) : A \sim \mathcal{D},\, y = g(A)\},
$
that is, $\mathcal{D}_g$ represents samples drawn from $\mathcal{D}$ whose labels are generated by $g$. 
Let $\mathcal{D}_{\mathrm{Unif}(\mathcal{C}_n)}$ denote the mixture distribution obtained by first sampling $g \sim \mathrm{Unif}(\mathcal{C}_n)$ and then setting $y = g(A)$.

Suppose $\mathcal{C}_n$ forms a $(1-\eta)$–pairwise independent function family: with probability at least $(1-\eta)$ over independently drawn graphs $A, A' \sim \mathcal{D}$,
whenever $\hat{c}_A \neq \hat{c}_{A'}$ and $[\hat{c}_A]_i = [c_A]_i \bmod 2$, 
the joint distribution $(g_S(A), g_S(A'))$ for a uniformly random $g_S \sim \mathcal{C}_n$ equals 
$\mathrm{Unif}(\{0,1\}^2)$. 
Under this condition, any $\SQ$ learner that distinguishes the target-labeled distribution 
$\mathcal{D}_{g^{*}}$ (for an unknown $g^{*}\in \mathcal{C}_n$) from the random mixture 
$\mathcal{D}_{\mathrm{Unif}(\mathcal{C}_n)}$ with classification error at most $1/4$ using queries of tolerance $\tau$ 
requires at least $\tfrac{\tau^{2}}{2\eta}$ statistical queries.

Finally, using the concentration properties of Erd\H{o}s–R\'enyi random graphs, 
we can show that the constructed class $\mathcal{C}_n$ satisfies 
$(1-\mathcal{O}(\exp(-n^{\Omega(1)})))$–pairwise independence. 
A detailed derivation can be found in~\cite[Appendix~E.1]{kiani2024hardness}.
\end{proofsketch}


In contrast, we present a $\QSQ$ algorithm that learns $\mathcal{C}_n$ with sample complexity polynomial in $n$, yielding a sharp separation between the classical
$\SQ$ and quantum $\QSQ$ learning models. At a high level, the algorithm proceeds in two stages: an efficient state preparation procedure that coherently computes
$c_A \bmod 2$, followed by a quantum Fourier sampling routine. To begin, we observe that each function in $\mathcal{C}_n$ can be expressed as  
\[
    g_{S}(A) = \hat{S}\cdot \hat{c}_{A} \bmod 2,
\]
where $[\hat{c}_{A}]_i = [c_{A}]_i \bmod 2$ and $[\hat{S}]_i = 1$ if $i \in S$ and $0$ otherwise.  Thus, learning a target function $g_S$ reduces to identifying the hidden binary string $\hat{S}$,  which can be efficiently recovered using quantum Fourier sampling~\cite{bernstein1993quantum}.  

We first define $x_A = \text{vec}(A)$ as the column-stacked vectorization of $A$. Since \(A\) is sampled uniformly from \(\{0,1\}^{n \times n}\), access to a \(\Qstat\) oracle provides quantum examples of the form
\[
    \ket{g_S}
    = \frac{1}{\sqrt{2^{n^2}}} \sum_{A \in \{0,1\}^{n \times n}}
    \ket{x_A}\ket{g_S(A)} .
\]
Therefore, in order to evaluate the function \(g_S(A)\) efficiently, it suffices to compute the corresponding input \(\hat{c}_{A}\) from the encoding \(\ket{x_A}\). The following lemma shows that there exists an efficient quantum circuit that prepares \(\hat{c}_{A}\) given access to the state \(\ket{g_S}\).
\begin{lmma}
    [State preparation]\label{lemma: state_prep} Let $\ket{g_{S}} = \frac{1}{\sqrt{N}}\sum_{A\in\{0,1\}^{n \times n}} \ket{x_A}\ket{g_{S}(A)}$, where $N=2^{n^2}$. There exists a $\text{poly}(n)$-gate quantum circuit $\mathsf{PREP}$ such that
    $$
        \mathsf{PREP} \ket{g_{S}}\ket{0^{\otimes n + n\lfloor\log(n)\rfloor + 1}} = \frac{1}{\sqrt{N}} \sum_{A\in\{0,1\}^{n \times n}}\ket{x_A} \ket{\hat{c}_{A}}\ket{g_{S}(A)}.
    $$
\end{lmma}
Next, we show that using a quantum Fourier sampling circuit, we could estimate the Fourier mass of any function $f$ on a subset of $\{0,1\}^m$ using a single query of $\Qstat$.
\begin{lmma}
\label{lmma:fourier-mass-one-query}
Let $f:\{0,1\}^m\to\{0,1\}$ and let $\ket{\psi_f}\;=\;2^{-n/2}\sum_{x\in\{0,1\}^m}\ket{x}\ket{f(x)}$ be the corresponding quantum example state, for any subset $T\subseteq\{0,1\}^m$, there is an observable $M_T$ with $\|M_T\|\le 1$ such that $
\Tr[M_T\ketbra{\psi_f}{\psi_f}] \;=\; \sum_{S\in T}\widehat{f}(S)^2$.
Consequently, a single call to the $\Qstat(M_T,\tau)$ returns a $\tau$-estimate of $\sum_{S\in T}\widehat{f}(S)^2$.
\label{lmma: QSQ_obs}
\end{lmma}
With these lemmas at hand, we can have the guarantee of $\QSQ$ algorithm as follows:
\begin{thrm}[Quantum Guarantee]
For any $\tau \in (0, \tfrac{1}{2})$, there exists a $\QSQ$ algorithm that exactly learns $\mathcal{C}_n$ using $\mathcal{O}(n)$ queries to $\Qstat(\tau)$, together with an additional $\mathrm{poly}(n)$ number of gates and ancilla qubits.  
\label{thrm: quantum_parity}
\end{thrm}
\begin{proofsketch} First, we show that the relevant symmetry statistics of the input can be coherently computed.
Given quantum examples $\ket{g_S}$, Lemma~\ref{lemma: state_prep} constructs a polynomial-size quantum circuit $\mathsf{PREP}$ that computes the degree-count vector modulo $2$, mapping
\[
\ket{g_S}\;\mapsto\;\ket{g'_S}
= \frac{1}{\sqrt{2^{n^2}}}\sum_A \ket{x_A}\ket{\hat c_A}\ket{g_S(A)}.
\]
This step preserves coherence and enables subsequent Fourier analysis on the hidden structure. Second, we exploit the fact that learning $g_S$ reduces to identifying the unknown binary string $\hat S$. And under the Walsh--Hadamard Fourier expansion, the Fourier influence $\mathrm{Inf}_i(g_S)$ equals $1$ if $i\in\mathrm{supp}(\hat S)$ and $0$ otherwise. By Lemma~\ref{lmma:fourier-mass-one-query}, for each coordinate $i$ there exists an efficiently implementable observable whose expectation equals $\mathrm{Inf}_i(g_S)$. Thus, a single $\QSQ$ query per coordinate suffices to distinguish whether $i\in\mathrm{supp}(\hat S)$, provided the tolerance satisfies $\tau<\tfrac12$.

Combining the Fourier-sampling observable with the state-preparation circuit yields a family of observables $\{O_i\}_{i=1}^{n+1}$, each implementable in polynomial time.
Querying all $O_i$ recovers $\hat S$ exactly using $\mathcal{O}(n)$ $\Qstat$ queries, completing the proof.
\end{proofsketch}
We defer the full proofs to Appendix~\ref{app: quantum_parity} for better readability.

\subsection{Learning general symmetries}\label{sec: general}

In the previous section, we established an exponential separation between quantum and classical models in learning permutation-invariant functions.
We next extend our analysis beyond permutation invariance to arbitrary symmetry groups. Specifically, we consider the problem of learning symmetric Boolean functions defined as follows:
\begin{defn}[Symmetric Boolean functions]\label{def: sym_func} Let input distribution $\mathcal{D}$ be uniform over $\mathcal{X} = \{0,1\}^{n}$. For a given symmetry group $G$ with 
    representation $\rho: G\to GL(\{0,1\}^n)$, the class of symmetric Boolean functions is defined as:
    \begin{equation}
        \mathcal{C}_{\rho} =\{ f: \{0,1\}^{n} \to \{0,1\}: \forall g \in G, \forall x \in \mathcal{X}: f(\rho(g)\cdot x) = f(x)\}.
    \end{equation} 
\end{defn}
To establish the hardness of $\mathcal{C}_{\rho}$, we consider the concept of the \emph{orbit representation} of a group action. Given a symmetry group \( G \) acting on the input space \( \mathcal{X} = \{0,1\}^n \) via a representation \( \rho: G \to \mathrm{GL}(\mathcal{X}) \), the action partitions the domain into \emph{orbits} under \( G \), defined as
\[
\mathcal{O}_\rho := \left\{ \mathcal{O}(x) := \{ \rho(g) \cdot x \mid g \in G \} \,\middle|\, x \in \mathcal{X} \right\}.
\]
Functions in the symmetric class \( \mathcal{C}_\rho \) are necessarily \emph{constant on orbits}. This structure directly induces correlations between function values: for inputs \( x, x' \) belonging to the same orbit, their function values are identical across all \( f \in \mathcal{C}_\rho \); for inputs in distinct orbits, their joint distribution over \( (f(x), f(x')) \) is uniform if $f \sim \text{Unif}(\mathcal{C}_\rho)$. Consequently, the \emph{pairwise correlation} between functions in \( \mathcal{C}_\rho \) is tightly governed by the \emph{distribution of orbit sizes}. Specifically, defining the orbit distribution \( p_{\mathcal{O}_\rho}(O) = |O| / |\mathcal{X}| \), we show that \( \mathcal{C}_\rho \) forms a \((1 - \|p_{\mathcal{O}_\rho}\|_2^2)\)-pairwise independent function family. This means that when the orbit distribution is nearly uniform (i.e., \( \|p_{\mathcal{O}_\rho}\|_2^2 \ll 1 \)), the function class exhibits \emph{low pairwise correlation}, which is a key indicator of learning hardness under classical statistical query models. In particular, we present this notion in the following lemma. For completeness, we include a proof in Appendix~\ref{app: pairwise}.
\begin{lmma}
 Consider $\Vert p_{\mathcal{O}_\rho}\Vert_2^2 = \sum_{O_k \in \mathcal{O}_{\rho}} (|O_k|/|\mathcal{X}|)^2$, with $(1- \Vert p_{\mathcal{O}_\rho}\Vert_2^2)$ over the choice of $x, x'\in \mathcal{X}$, for every $f\in \text{Unif}(\mathcal{C}_{\rho})$ the distribution of $(f(x), f(x'))$ is equal to uniform distribution over $\{0, 1\}^{\otimes 2}$.
    \label{lmma: H_pairwise}
\end{lmma}
Based on Lemma~\ref{lmma: H_pairwise}, Kiani \textit{et. al.}~\cite{kiani2024hardness} established a lower bound $
\Omega\!\left({\tau^2}/{\Vert p_{\mathcal{O}_\rho}\Vert_2^{2}}\right)$ on the query complexity of learning under the classical $\SQ$ model, showing that the hardness of learning $\mathcal{C}_{\rho}$ is governed by the orbit distribution induced by the group action. In particular, for the cyclic group the number of distinct orbits is $\Theta(2^n/n)$, yielding $\Vert p_{\mathcal{O}_\rho}\Vert_2^{2}=\Theta(n/2^n)$ and an exponential lower bound. We therefore ask whether the $\QSQ$ model, by allowing non-diagonal observables, can exploit orbit structure beyond what is achievable with diagonal (classical) queries, or whether the same orbit-induced constraints fundamentally limit quantum learnability.

\paragraph{Diagonal Observables.} We first restrict to the set of diagonal observables $O$. Note that the $\QSQ$ model with access only to diagonal observables reduces to the classical $\SQ$ model. Our goal is to recover the classical hardness result of~\cite{kiani2024hardness} within the $\QSD$ lower bound. 

To do so, we employ Lemma~\ref{lmma: variance_bound}. Given the $\Qstat$ accessing quantum example state of $$\ket{\psi_f}
=\sum_{x\in\mathcal{X}}\frac{1}{\sqrt{|\mathcal{X}|}}\ket{x,f(x)},$$ for $f\in \mathcal{C}_\rho$, we have:
\begin{equation}
    \mathrm{Var}(\mathcal{C}_\rho) = \mathbb{E}_f[\Tr[O\ket{\psi_f}\bra{\psi_f}]^2] - (\mathbb{E}_f[\Tr[O\ket{\psi_f}\bra{\psi_f}]])^2
\end{equation}
We consider diagonal observables $O$ of the form $O = \sum_{x\in \mathcal{X}} \phi(x, f(x))\ket{x, f(x)}\bra{x, f(x)}$ for $\phi: \{0, 1\}^n \times \{0, 1\}\to [-1, 1]$ as $\Vert O \Vert \leq 1$. Thus, we have: $$\Tr[O\ket{\psi_f}\bra{\psi_f}] = \frac{1}{{|\mathcal{X}|}}\sum_x  \phi(x, f(x)) := \Phi_f$$
Then, we can write $\mathrm{Var}(\mathcal{C}_\rho)$ as follows:
\begin{align}
    \mathrm{Var}(\mathcal{C}_\rho)
    &= \mathbb{E}_f[\Phi^2_f] - \mathbb{E}_f[\Phi_f]\cdot \mathbb{E}_{f'}[\Phi_{f'}]\\
    &= \mathbb{E}_{f, f'}[\Phi_f\cdot \Phi_f - \Phi_f\cdot \Phi_{f'}] \\
    &=  \frac{1}{{|\mathcal{X}|^2}}\sum_{x, x'}\left(\mathbb{E}_{f}[\phi(x, f(x))\cdot\phi(x', f(x'))] - \mathbb{E}_{f}[\phi(x, f(x))]\cdot\mathbb{E}_{f'}[\phi(x', f'(x'))]\right)
\end{align}
Using Lemma~\ref{lmma: H_pairwise}, we have with 
$(1- \Vert p_{\mathcal{O}_\rho}\Vert_2^2)$ over the choice of $x, x'\in \mathcal{X}$, the distribution of $(f(x), f(x'))$ will act as $\text{Unif}(\{0, 1\})^{\otimes 2}$. In such case, 
$$
 \mathbb{E}_{f}[\phi(x, f(x))\cdot\phi(x', f(x'))] = \mathbb{E}_{y, y' \in \{0, 1\}} \phi(x, y)\cdot\phi(x', y') = \mathbb{E}_{f}[\phi(x, f(x))]\cdot\mathbb{E}_{f'}[\phi(x', f'(x'))]
$$
the second equality holds because $f, f'$ are drawn independently from $\mathcal{C}_{\rho}$. Thus, with probability of $(1- \Vert p_{\mathcal{O}_\rho}\Vert_2^2)$, this term vanishes. Therefore,
\begin{equation}
\mathrm{Var}(\mathcal{C}_\rho) = \Theta(\Vert p_{\mathcal{O}_\rho}\Vert_2^2)
\label{eqn: diagonal_variance}
\end{equation}
From this analysis, we can obtain the $\QSQ$ lower bound for the diagonal $O$ as follows:
\begin{thrm}\label{thrm: qsq_lb_diagonal}
    Let $\tau < 0.096$ and $O$ be any diagonal observable with $\Vert O \Vert \leq 1$. for every $\QSQ$ algorithm restricted to such observables that learn $\mathcal{C}_\rho$ to trace distance below $0.05$ with probability of 0.99, the number of queries to $\Qstat(\tau)$ is at least
    $$
     \Omega\left(\frac{\tau^2}{\Vert p_{\mathcal{O}_\rho}\Vert_2^2}\right). 
    $$
\end{thrm}
\begin{proof}
    We first convert the learning problem to a decision problem between $\rho_f = \ket{\psi_f}\bra{\psi_f}$ where $f$ is any state in $\mathcal{C}_\rho$, and an adversarial state $\sigma = \mathbb{E}_{f\sim \mathcal{C}_\rho}[\ket{\psi_f}\bra{\psi_f}]$. Using Lemma~\ref{lmma:qsd_bound_qsq}, we have
    $$
    \QSQ^{\varepsilon,\delta}_{\tau}(\mathcal{C}_\rho) \geq (1-2\delta)\QSD_{\tau}(\mathcal{C}_\rho, \sigma) -1.
    $$
    To apply this relationship, we need $\min_f \TRD(\rho_f, \sigma) > 2(\tau +\epsilon)$. We observe that:
    $$
        \TRD(\rho_f, \sigma) \geq 1 - \sqrt{\mathbb{E}_{f'} [|\braket{\psi_f}{\psi_{f'}}|^2]} \geq 1-\sqrt{\frac{1}{2}}.
    $$
    Fixing $\epsilon = 0.05$ and $\delta =  0.01$, then the Lemma~\ref{lmma:qsd_bound_qsq} holds if $\tau < 0.096$. Along with Lemma~\ref{lmma: variance_bound} and~\eqref{eqn: diagonal_variance}, we complete the proof.
\end{proof}
This is unsurprisingly equal to the bound in the classical $\SQ$ model \cite{kiani2024hardness}. However, the power of $\QSQ$ models comes from their capability to use entanglement induced from non-diagonal observable $O$. Thus, we are rather interested in the general design of $O$, where we could utilize entanglement power.

\paragraph{General Observables.} In this case, we still use the variance bound method to lower bound $\QSD$, thereby $\QSQ$. The lower bound is presented in the following theorem
\begin{thrm}\label{thrm: qsq_lb_general}
   Let $\tau < 0.096$. There exists a family of observables $O$ such that $\Vert O \Vert \leq 1$ and for every $\QSQ$ algorithm restricted to observables in such family that learns $\mathcal{C}_\rho$ to the trace distance below $0.05$ with probability of $0.99$, the number of queries to $\Qstat(\tau)$ is at least
    $$
     \Omega\left(\frac{\tau^2 |\mathcal{X}|}{\max_{O_k \in \mathcal{O}_\rho}|O_k|}\right). 
    $$
\end{thrm}
\begin{proofsketch}
Full details of the proof can be found in Appendix~\ref{app: proof_general_obs}. Following the same overall strategy as Theorem~\ref{thrm: qsq_lb_diagonal}, we extend the argument to handle a general observable $O$. For an arbitrary observable $O$, we analyze $\mathrm{Var}_f(\Tr[O\ket{\psi_f}\bra{\psi_f}])$. Using the phase-kickback trick, we expand $\ket{\psi_f}\bra{\psi_f}$ in a basis that separates the uniform component from the phase-encoded component and decompose $O$ accordingly. This reduces the variance calculation to a sum of terms involving products of matrix elements over group orbits.

On the other hand, by exploiting the pairwise independence of $\mathcal{C}_\rho$ (Lemma~\ref{lmma: H_pairwise}), most mixed terms vanish in expectation unless the corresponding inputs lie in the same orbit. As a result, all the related terms from different orbits cancel, and only terms
associated with orbit-wise structure remain. These surviving terms can be bounded using Cauchy–Schwarz and the operator-norm constraint on $O$, yielding
$
\mathrm{Var}(\mathcal{C}_\rho)\;\le\; \frac{\max_{O_k \in \mathcal{O}_\rho} |O_k|}{|\mathcal{X}|},
$
where $O_k$ ranges over the orbits of the symmetry group. Then, we design a family of $O$ such that this inequality is tight, as such 
\begin{equation}\label{eqn: non_diagonal_variance}
    \mathrm{Var}(\mathcal{C}_\rho) =  \Theta\left(\frac{\max_{O_k \in \mathcal{O}_\rho} |O_k|}{|\mathcal{X}|}\right).
\end{equation}
Finally, we apply Lemma~\ref{lmma: variance_bound} to complete the proof.
\end{proofsketch}

\paragraph{Discussion.} Theorem~\ref{thrm: qsq_lb_general} extends the $\QSQ$ lower bound to arbitrary observables, while Theorem~\ref{thrm: qsq_lb_diagonal} recovers the classical $\SQ$ hardness of~\cite{kiani2024hardness}. A direct comparison yields
\[
\|p_{\mathcal{O}_\rho}\|_2^2
=
\sum_{O_k \in \mathcal{O}_\rho} \frac{|O_k|^2}{|\mathcal{X}|^2}
\;\le\;
\frac{\max_{O_k \in \mathcal{O}_\rho}|O_k|}{|\mathcal{X}|},
\]
where we use $\sum_{O_k \in \mathcal{O}_\rho} |O_k| = |\mathcal{X}|$. Combined with~\eqref{eqn: diagonal_variance} and~\eqref{eqn: non_diagonal_variance}, this shows that allowing general $\Qstat$ queries can increase the variance of statistical estimates relative to the classical $\Stat$ model. Since variance controls the correlation strength accessible to statistical queries, this amplification improves statistical distinguishability. However, for commonly studied symmetry groups, the above inequality is tight up to constants, implying that the $\QSQ$ lower bound cannot asymptotically exceed its classical $\SQ$ counterpart. In contrast, when the orbit distribution is highly skewed, the gap between $\|p_{\mathcal{O}_\rho}\|_2^2$ and $\max_{O_k \in \mathcal{O}_\rho}|O_k|/|\mathcal{X}|$ can grow, allowing non-diagonal observables to induce substantially larger variance. These are formally stated in the following remarks.
\begin{remark}[Matching lower bounds under orbit regularity]\label{rem: common_groups}
For group actions satisfying
$\|p_{\mathcal{O}_\rho}\|_2^2=\Theta(\max_{O_k \in \mathcal{O}_\rho}|O_k|/|\mathcal{X}|)$,
including cyclic shifts and other orbit-regular symmetries, the $\QSQ$ and classical $\SQ$ query complexity lower bounds match up to constant factors.
\end{remark}
\begin{remark}[Advantage under skewed orbit distributions]\label{rem: skewed_orbit}
If a group action has an orbit $O^\star$ of size $|O^\star|=\sqrt{|\mathcal{X}|}$ and all remaining orbits has the size of $1$, then
$
\|p_{\mathcal{O}_\rho}\|_2^2=\Theta(1/|\mathcal{X}|),
$ and $
\max_{O_k \in \mathcal{O}_\rho}|O_k|/|\mathcal{X}|=\Theta(1/\sqrt{|\mathcal{X}|}),
$
so the $\Stat$ yield variance $\Theta(1/|\mathcal{X}|)$, while the $\Qstat$ can achieve variance $\Theta(1/\sqrt{|\mathcal{X}|})$.
\end{remark}
However, it remains unclear to what extent such skewed orbit structures arise in practical learning problems. Identifying natural data distributions or symmetry groups that violate the orbit-regularity condition, and determining whether they can be exploited algorithmically, remains an open question.



\subsection{Possible Advantages on Tolerance Parameter}\label{sec: tolerance}
While the variance bound provides a powerful tool for establishing general lower bounds on $\QSQ$ learning, it does not capture potential advantages that may arise through the tolerance parameter. In this section, we identify a regime where the tolerance parameter $\tau$ creates a strict separation between quantum and classical learnability. We show that for function classes with specific orbit structures, there exists a window of tolerance values where classical $\SQ$ learners necessarily fail, yet $\QSQ$ learners can efficiently identify the target function. We present the full proof in Appendix~\ref{app: tolerance_separation}.

\begin{thrm}[Tolerance-Based Quantum Advantage]\label{thm:tolerance_separation}
Let $\mathcal{C}_{\rho}$ be a class of symmetric Boolean functions under group representation $\rho$. Let $\mathcal{S} \subseteq \mathcal{C}_{\rho}$ be a subclass of functions such that every $f \in \mathcal{S}$ is non-zero on a set of inputs $S_f \subset \mathcal{X}$ with measure $\mu(S_f) = \frac{|S_f|}{|\mathcal{X}|} = \zeta$. Assume that for any distinct $f, g \in \mathcal{S}$, their supports are disjoint ($S_f \cap S_g = \emptyset$).

If the orbit size fraction $\zeta$ satisfies $\zeta < 0.4$, then there exists a non-empty interval of tolerance parameters $\tau \in (\tau_{\SQ}, \tau_{QSQ})$ where:
\begin{enumerate}
    \item \textbf{$\SQ$ Failure:} No classical $\SQ$ algorithm can learn $\mathcal{S}$ with tolerance $\tau$.
    \item \textbf{$\QSQ$ Success:} There exists a $\QSQ$ algorithm that learns $\mathcal{S}$ with tolerance $\tau$ using $O(\log |\mathcal{S}|)$ queries.
\end{enumerate}
Specifically, the separation holds for any $\tau$ such that $
2\zeta < \tau < \sqrt{2\zeta - \zeta^2}.
$
\end{thrm}
\begin{proofsketch}
For functions in $\mathcal{S}$, which differ from the zero function only on a fraction $\zeta$ of the domain, every such expectation differs from that of the zero
function by at most $2\zeta$. Consequently, when the query tolerance $\tau$ exceeds this scale, the $\SQ$ oracle can respond consistently as if the target
were the zero function, rendering all functions in $\mathcal{S}$ indistinguishable and learning impossible.

In contrast, the $\QSQ$ model allows access to global quantum observables. The
quantum example states $\ket{\psi_f}$ and $\ket{\psi_0}$ corresponding to a
function $f$ and the zero function have overlap
$\braket{\psi_0}{\psi_f} = 1-\zeta$, which implies a trace distance
\[
d_{\mathrm{tr}}(\ket{\psi_f},\ket{\psi_0})
= \sqrt{1 - |\langle \psi_f \mid \psi_0 \rangle|^2}
= \sqrt{2\zeta - \zeta^2}.
\]
By the optimality of the Helstrom measurement, there exists an observable whose
expectation values on these two states differ by $ \sqrt{2\zeta - \zeta^2}$.
Therefore, as long as the tolerance $\tau$ is below this trace-distance scale,
a $\QSQ$ learner can reliably distinguish $f$ from the zero function.

Since the classical indistinguishability threshold scales as $2\zeta$
while the quantum distinguishability threshold scales as $ \sqrt{2\zeta - \zeta^2}$,
there exists a broad regime of $\tau$ for which classical $\SQ$ learning fails
but $\QSQ$ learning succeeds, yielding the claimed separation.
\end{proofsketch}

\section{Conclusion and Open Problems}\label{sec: conclusion}
We studied the role of symmetry in quantum learning within the quantum statistical query ($\QSQ$) framework, characterizing when symmetry enables a quantum advantage and when it collapses to classical hardness. We exhibited an exponential separation between $\QSQ$ and $\SQ$ learners for permutation-invariant functions, proved matching query-complexity lower bounds for most commonly studied symmetries, and identified highly skewed orbit distributions as a source of asymptotic advantage. We also established a tolerance-based regime in which quantum learners succeed at noise levels that defeat classical $\SQ$ algorithms. Together, these results delineate both the power and the limitations of quantum learning under symmetry constraints.

At the same time, our results highlight several important limitations. First, the $\QSQ$ framework assumes access to prepared quantum examples, yet the efficient and scalable preparation of such states is itself a nontrivial task~\cite{rattew2022preparing, rosenkranz2024quantum}. In practice, this preparation cost may offset theoretical advantages; for example, the spacetime volume of state-loading circuits is expected to scale exponentially with the system size~\cite{jaques2023qram, aaronson2014quantum}. Second, although our analysis identifies highly skewed orbit distributions as a condition for an asymptotic quantum advantage, it remains unclear whether such symmetry structures arise naturally in practical learning problems. Third, while the tolerance-based separation demonstrates a form of quantum advantage, realizing this advantage in practice typically requires implementing the Helstrom measurement, which is generally difficult to perform efficiently~\cite{han2020helstrom}. Finally, our analysis is carried out under the uniform input distribution, which provides a clean and tractable setting for establishing lower bounds. Extending these results to more general distributions remains an open problem, and it is unclear whether the same separation or collapse phenomena persist beyond the uniform case.

\paragraph{Acknowledgments }{We thank Thinh Le, Ryan Mann, and Cl\'ement Canonne for their constructive feedback. TN is supported by a scholarship from the  Sydney Quantum Academy, PHDR06031.}

\printbibliography

@misc{arunachalam2023roleentanglementstatisticslearning,
      title={On the Role of Entanglement and Statistics in Learning}, 
      author={Srinivasan Arunachalam and Vojtech Havlicek and Louis Schatzki},
      year={2023},
      eprint={2306.03161},
      archivePrefix={arXiv},
      primaryClass={quant-ph},
      url={https://arxiv.org/abs/2306.03161}, 
}

@article{schatzki2024theoretical,
  title={Theoretical guarantees for permutation-equivariant quantum neural networks},
  author={Schatzki, Louis and Larocca, Martin and Nguyen, Quynh T and Sauvage, Frederic and Cerezo, Marco},
  journal={npj Quantum Information},
  volume={10},
  number={1},
  pages={12},
  year={2024},
  publisher={Nature Publishing Group UK London}
}

@article{kearns1998efficient,
  title={Efficient noise-tolerant learning from statistical queries},
  author={Kearns, Michael},
  journal={Journal of the ACM (JACM)},
  volume={45},
  number={6},
  pages={983--1006},
  year={1998},
  publisher={ACM New York, NY, USA}
}

@article{bronstein2021geometric,
  title={Geometric deep learning: Grids, groups, graphs, geodesics, and gauges},
  author={Bronstein, Michael M and Bruna, Joan and Cohen, Taco and Veli{\v{c}}kovi{\'c}, Petar},
  journal={arXiv preprint arXiv:2104.13478},
  year={2021}
}

@inproceedings{chen2022learning,
  title={Learning deep relu networks is fixed-parameter tractable},
  author={Chen, Sitan and Klivans, Adam R and Meka, Raghu},
  booktitle={2021 IEEE 62nd Annual Symposium on Foundations of Computer Science (FOCS)},
  pages={696--707},
  year={2022},
  organization={IEEE}
}

@article{long2019generalization,
  title={Generalization bounds for deep convolutional neural networks},
  author={Long, Philip M and Sedghi, Hanie},
  journal={arXiv preprint arXiv:1905.12600},
  year={2019}
}

@article{pesah2021absence,
  title={Absence of barren plateaus in quantum convolutional neural networks},
  author={Pesah, Arthur and Cerezo, Marco and Wang, Samson and Volkoff, Tyler and Sornborger, Andrew T and Coles, Patrick J},
  journal={Physical Review X},
  volume={11},
  number={4},
  pages={041011},
  year={2021},
  publisher={APS}
}

@article{anschuetz2022quantum,
  title={Quantum variational algorithms are swamped with traps},
  author={Anschuetz, Eric R and Kiani, Bobak T},
  journal={Nature Communications},
  volume={13},
  number={1},
  pages={7760},
  year={2022},
  publisher={Nature Publishing Group UK London}
}

@article{tahmasebi2023exact,
  title={The Exact Sample Complexity Gain from Invariances for Kernel Regression on Manifolds},
  author={Tahmasebi, Behrooz and Jegelka, Stefanie},
  journal={arXiv preprint arXiv:2303.14269},
  year={2023}
}

@article{reyzin2020statistical,
  title={Statistical queries and statistical algorithms: Foundations and applications},
  author={Reyzin, Lev},
  journal={arXiv preprint arXiv:2004.00557},
  year={2020}
}

@inproceedings{mei2021learning,
  title={Learning with invariances in random features and kernel models},
  author={Mei, Song and Misiakiewicz, Theodor and Montanari, Andrea},
  booktitle={Conference on Learning Theory},
  pages={3351--3418},
  year={2021},
  organization={PMLR}
}

@article{batzner20223,
  title={E (3)-equivariant graph neural networks for data-efficient and accurate interatomic potentials},
  author={Batzner, Simon and Musaelian, Albert and Sun, Lixin and Geiger, Mario and Mailoa, Jonathan P and Kornbluth, Mordechai and Molinari, Nicola and Smidt, Tess E and Kozinsky, Boris},
  journal={Nature communications},
  volume={13},
  number={1},
  pages={2453},
  year={2022}
}

@misc{arunachalam2020quantumstatisticalquerylearning,
      title={Quantum statistical query learning}, 
      author={Srinivasan Arunachalam and Alex B. Grilo and Henry Yuen},
      year={2020},
      eprint={2002.08240},
      archivePrefix={arXiv},
      primaryClass={quant-ph},
      url={https://arxiv.org/abs/2002.08240}, 
}

@article{kiani2024hardness,
  title={On the hardness of learning under symmetries},
  author={Kiani, Bobak T and Le, Thien and Lawrence, Hannah and Jegelka, Stefanie and Weber, Melanie},
  journal={arXiv preprint arXiv:2401.01869},
  year={2024}
}

@inproceedings{bernstein1993quantum,
  title={Quantum complexity theory},
  author={Bernstein, Ethan and Vazirani, Umesh},
  booktitle={Proceedings of the twenty-fifth annual ACM symposium on Theory of computing},
  pages={11--20},
  year={1993}
}

@misc{draper2000additionquantumcomputer,
      title={Addition on a Quantum Computer}, 
      author={Thomas G. Draper},
      year={2000},
      eprint={quant-ph/0008033},
      archivePrefix={arXiv},
      primaryClass={quant-ph},
      url={https://arxiv.org/abs/quant-ph/0008033}, 
}

@article{angrisani2022quantum,
  title={Quantum local differential privacy and quantum statistical query model},
  author={Angrisani, Armando and Kashefi, Elham},
  journal={arXiv preprint arXiv:2203.03591},
  year={2022}
}

@inproceedings{wu2022graph,
  title={Graph neural networks: foundation, frontiers and applications},
  author={Wu, Lingfei and Cui, Peng and Pei, Jian and Zhao, Liang and Guo, Xiaojie},
  booktitle={Proceedings of the 28th ACM SIGKDD conference on knowledge discovery and data mining},
  pages={4840--4841},
  year={2022}
}

@inproceedings{cohen2016group,
  title={Group equivariant convolutional networks},
  author={Cohen, Taco and Welling, Max},
  booktitle={International conference on machine learning},
  pages={2990--2999},
  year={2016},
  organization={PMLR}
}

@article{West_2024,
   title={Provably Trainable Rotationally Equivariant Quantum Machine Learning},
   volume={5},
   ISSN={2691-3399},
   url={http://dx.doi.org/10.1103/PRXQuantum.5.030320},
   DOI={10.1103/prxquantum.5.030320},
   number={3},
   journal={PRX Quantum},
   publisher={American Physical Society (APS)},
   author={West, Maxwell T. and Heredge, Jamie and Sevior, Martin and Usman, Muhammad},
   year={2024},
   month=jul }

@article{lecun2002gradient,
  title={Gradient-based learning applied to document recognition},
  author={LeCun, Yann and Bottou, L{\'e}on and Bengio, Yoshua and Haffner, Patrick},
  journal={Proceedings of the IEEE},
  volume={86},
  number={11},
  pages={2278--2324},
  year={2002},
  publisher={Ieee}
}

@article{mcclean2018barren,
  title={Barren plateaus in quantum neural network training landscapes},
  author={McClean, Jarrod R and Boixo, Sergio and Smelyanskiy, Vadim N and Babbush, Ryan and Neven, Hartmut},
  journal={Nature communications},
  volume={9},
  number={1},
  pages={4812},
  year={2018},
  publisher={Nature Publishing Group UK London}
}

@article{Cerezo_2025,
   title={Does provable absence of barren plateaus imply classical simulability?},
   volume={16},
   ISSN={2041-1723},
   url={http://dx.doi.org/10.1038/s41467-025-63099-6},
   DOI={10.1038/s41467-025-63099-6},
   number={1},
   journal={Nature Communications},
   publisher={Springer Science and Business Media LLC},
   author={Cerezo, M. and Larocca, Martin and García-Martín, Diego and Diaz, N. L. and Braccia, Paolo and Fontana, Enrico and Rudolph, Manuel S. and Bermejo, Pablo and Ijaz, Aroosa and Thanasilp, Supanut and Anschuetz, Eric R. and Holmes, Zoë},
   year={2025},
   month=aug }

@inproceedings{aharonov2003adiabatic,
  title={Adiabatic quantum state generation and statistical zero knowledge},
  author={Aharonov, Dorit and Ta-Shma, Amnon},
  booktitle={Proceedings of the thirty-fifth annual ACM symposium on Theory of computing},
  pages={20--29},
  year={2003}
}

@article{rattew2022preparing,
  title={Preparing arbitrary continuous functions in quantum registers with logarithmic complexity},
  author={Rattew, Arthur G and Koczor, B{\'a}lint},
  journal={arXiv preprint arXiv:2205.00519},
  year={2022}
}

@article{rosenkranz2024quantum,
  title={Quantum state preparation for multivariate functions},
  author={Rosenkranz, Matthias and Brunner, Eric and Marin-Sanchez, Gabriel and Fitzpatrick, Nathan and Dilkes, Silas and Tang, Yao and Kikuchi, Yuta and Benedetti, Marcello},
  journal={arXiv preprint arXiv:2405.21058},
  year={2024}
}

@article{jaques2023qram,
  title={QRAM: A survey and critique},
  author={Jaques, Samuel and Rattew, Arthur G},
  journal={arXiv preprint arXiv:2305.10310},
  year={2023}
}

@article{aaronson2014quantum,
  title={Quantum machine learning algorithms: Read the fine print},
  author={Aaronson, Scott},
  journal={Nature Physics},
  pages={5},
  year={2014}
}

@article{han2020helstrom,
  title={Helstrom measurement: A nondestructive implementation},
  author={Han, Rui and Leuchs, Gerd and Bergou, J{\'a}nos A},
  journal={Physical Review A},
  volume={101},
  number={3},
  pages={032103},
  year={2020},
  publisher={APS}
}

\appendix

\section{Extended background of Quantum Statistical Learning Framework}\label{app:qsd_reductions}
We first introduce the quantum statistical query $(\QSQ)$ learning model following the definitions and notations in \cite{arunachalam2023roleentanglementstatisticslearning, arunachalam2020quantumstatisticalquerylearning}. Extended from the classical statistical query model ($\SQ$) \cite{kearns1998efficient}, $\QSQ$ brings simple yet powerful tools to understand the advantage of quantum examples in machine learning \cite{arunachalam2020quantumstatisticalquerylearning}.  In the classical $\SQ$ framework, a concept class is given by functions $\mathcal{C} \subseteq \{ c: \mathcal{X} \to \{0,1\} \}$. The goal is to identify an unknown target concept $c^\ast \in \mathcal{C}$ under an unknown distribution $\mathcal{D}$ on $\mathcal{X}$. Rather than observing labeled samples $(x,c^\ast(x))$, the learner interacts with a \emph{statistical query oracle} $\Stat$. For a query function $\phi : \mathcal{X} \times \{0,1\} \to [-1,1]$ and a tolerance parameter $\tau > 0$, the oracle returns
\[
\Stat(\phi,\tau) = \alpha 
\quad \text{such that} \quad
\Big| \alpha - \mathbb{E}_{x \sim \mathcal{D}}[\phi(x,c^\ast(x))] \Big| \leq \tau.
\]
The learner adaptively chooses a sequence of $m$ queries $\{(\phi_i,\tau_i)\}_{i=1}^m$ and, based on the responses, outputs a hypothesis $h : \mathcal{X} \to \{0,1\}$ satisfying a classification error $\varepsilon$: $\PR_{x \sim \mathcal{D}}[ h(x) \neq c^\ast(x) ] \leq \varepsilon$. Within the $\SQ$ framework, the complexity is measured not by raw samples but by the number of queries $\Stat$. With a slight abuse of notation, we denote by $\Stat(\tau)$ a query to any map $\phi$ satisfying $\phi : \mathcal{X} \times \{0,1\} \to [-1,1]$.

The quantum statistical query ($\QSQ$) model generalizes this framework to the quantum setting. Instead of classical labeled examples, information about the target concept is encoded in the \emph{quantum example state}
\[
\ket{\psi_{c^\ast}} = \sum_{x \in \mathcal{X}} \sqrt{\mathcal{D}(x)} \; \ket{x,c^\ast(x)}.
\]
Same as classical $\SQ$ models, a $\QSQ$ learner does not directly manipulate these states, but through a \emph{quantum statistical oracle} $\Qstat$, defined as follows.

\begin{defn}[Quantum Statistical Oracle, restate Definition~\ref{defn: qstat}]
Let $O \in \mathbb{C}^{2^{n+1} \times 2^{n+1}}$ be any observable with operator norm $\|O\| \leq 1$, and let $\tau > 0$.  
For a target concept $c^\ast \in \mathcal{C}$ encoded in the state $\ket{\psi_{c^\ast}}$, the quantum statistical oracle returns
\[
\Qstat(O,\tau) = \alpha \quad \text{such that} \quad 
\big| \alpha - \langle \psi_{c^\ast} | O | \psi_{c^\ast} \rangle \big| \leq \tau.
\]
With a slight abuse of notation, we denote by $\Qstat(\tau)$ a query to any observable $O$ satisfying $\|O\| \leq 1$.
\end{defn}

A crucial aspect of the $\QSQ$ model is that the learner itself remains a {classical} randomized algorithm. It never manipulates quantum states directly, but only receives approximate numerical values $\alpha$ returned by the oracle. In this sense, the $\QSQ$ model is a precise mathematical analogue of the $\SQ$ model, with expectations over labeled examples replaced by expectation values of observables:
\[
\underbrace{\mathbb{E}_{(x,c^\ast(x))}[\phi(x,c^\ast(x))]}_{\text{$\SQ$ framework}}
\quad \longleftrightarrow \quad 
\underbrace{\operatorname{Tr}(O\ket{\psi_c^\ast}\bra{\psi_c^\ast})}_{\text{$\QSQ$ framework}}.
\]
The key distinction in the quantum setting comes from the design of the observable $O$. Note that, when $O$ is diagonal with entries given by $\phi(x,c^\ast(x))$, the $\QSQ$ oracle reduces to the classical $\SQ$ oracle, since
$
\langle \psi_{c^\ast} | O | \psi_{c^\ast} \rangle 
= \mathbb{E}_{x \sim \mathcal D}[\phi(x,c^\ast(x))].
$ Hence, no quantum advantage is possible. The only potential source of quantum advantage 
in the $\QSQ$ framework comes from the ability to query with non-diagonal observables $O$ that enable access to non-classical correlations. This allows the quantum learner to extract correlations accessible to any classical $\SQ$ procedure, and therefore represents the precise point where quantum resources can outperform classical ones.

Throughout this paper, we focus on lower bounds on the query complexity of the $\QSQ$ model, which is the number of oracle calls $\Qstat(O,\tau)$ required to complete a learning task, and on how this complexity compares to that of the classical $\SQ$ model on the same problem. Establishing such lower bounds, however, is highly non-trivial. Earlier approaches based on differential privacy and communication-complexity arguments yielded only linear lower bounds in the input size $n$, far weaker than the hardness expected for many fundamental learning problems~\cite{angrisani2022quantum, arunachalam2020quantumstatisticalquerylearning}. To overcome this limitation, Arunachalam, Havlíček, and Schatzki introduced the \textit{Quantum Statistical Dimension} ($\QSD$)~\cite{arunachalam2023roleentanglementstatisticslearning}, which is a combinatorial parameter that extends the classical statistical dimension to the quantum setting and provides a unified framework for proving stronger $\QSQ$ lower bounds. Leveraging $\QSD$ alongside variance- and correlation-based arguments, they established tight bounds for several canonical tasks, including purity testing, shadow tomography, learning coset states, planted biclique states, and degree-2 function classes. In the following section, we formally define $\QSD$ and describe two key techniques that enable us to derive lower bounds on $\QSD$, and consequently on the query complexity of the $\QSQ$ model.

The standard methodology for proving lower bounds of $\QSQ$ is to reduce a learning problem to an associated \emph{decision problem}, and then to analyze the hardness of the decision problem using combinatorial parameters. Central among these parameters is the quantum statistical dimension ($\QSD$), which extends the classical notion of statistical dimension to the quantum setting. Intuitively, $\QSD$ quantifies the number of observables that are needed to reliably distinguish states from a given concept class, and hence provides a direct measure of the complexity of learning in the $\QSQ$ model. Beyond $\QSD$, two key techniques based on variance analysis and average correlation are employed to obtain sharper estimates and to capture the structural properties of the concept class under consideration. 

Formally, let $\mathcal{C}$ be a concept class of $n$-qubit quantum states, $\sigma$ be an $n$-qubit state with $\sigma \notin \mathcal{C}$, and $\tau > 0$ a query tolerance, we consider a quantum statistical decision problem for $(\mathcal{C}, \sigma)$ defined as: for an unknown state $\rho$, given $\Qstat(\tau)$ access to $\rho$ decide if $\rho \in \mathcal{C} $ or $\rho = \sigma$. Shown in~\cite{arunachalam2023roleentanglementstatisticslearning}, the complexity of this decision task is captured by the \emph{Quantum Statistical Dimension} ($\QSD$), which is formally defined as follows:

\begin{defn}[Quantum Statistical Dimension ($\QSD$), restate Definition~\ref{defn: qsd}]
For a distribution $\mu$ over $\mathcal{C}$ and $\sigma \notin \mathcal{C}$ be an $n$-qubit state, define the \emph{maximum covered fraction}
\[
   \kappa_{\tau}\text{-frac}(\mu,\sigma) \;=\;
   \max_{\|O\|\leq 1} \PR_{\rho \sim \mu}\!\left[ \, \big| \Tr[O(\rho - \sigma)]\big| > \tau \,\right],
\]
where the maximization is over all observables $O$ with operator norm $\|O\|\leq 1$. 
The \emph{quantum statistical dimension} is
\[
   \QSD_{\tau}(\mathcal{C},\sigma) \; :=\; 
   \sup_{\mu}\, \bigl[ \kappa_{\tau}\text{-frac}(\mu,\sigma) \bigr]^{-1}.
\]
\end{defn}
Intuitively, $\QSD_{\tau}(\mathcal{C},\sigma)$ measures the smallest expected number of observables needed to reliably distinguish any $\rho \in \mathcal{C}$ from $\sigma \notin \mathcal{C}$ using tolerance $\tau$. A large $\QSD$ implies that many observables are required, which translates into a high $\QSQ$ complexity for learning.

The connection between $\QSD$ and $\QSQ$ learning is established in two steps. First, one shows that \emph{learning a concept class is at least as hard as deciding membership in that class}. This means that if we can learn the class $\mathcal{C}$, then, in particular, we can solve the decision problem of 
determining whether an unknown state belongs to $\mathcal{C}$ or equals some fixed state $\sigma \notin \mathcal{C}$.

\begin{lmma}[Learning is as hard as deciding~\cite{arunachalam2023roleentanglementstatisticslearning}]\label{lmma: learn_as_decide}
Let $\varepsilon \geq \tau > 0$ and suppose $\min_{\rho \in \mathcal{C}} d_{\mathrm{tr}}(\rho,\sigma) > 2(\tau+\varepsilon)$. Furthermore, we denote $\QSQ^{\varepsilon,\delta}_{\tau}(\mathcal{C})$ as the number of $\Qstat(\tau)$ queries made by a $\QSQ$ algorithm that on input $\rho$ outputs $\pi \in \mathcal C$ such that $\TRD(\pi, \rho) \leq \epsilon$ with probability $\geq 1-\delta$. Then
\[
   \QSQ^{\varepsilon,\delta}_{\tau}(\mathcal{C}) \;\ge\; 
   \QQC^{\delta}_{\tau}(\mathcal{C},\sigma) - 1,
\]
where $\QQC^{\delta}_{\tau}(\mathcal{C},\sigma)$ is the minimum number of $\Qstat(\tau)$ queries needed to decide whether $\rho \in \mathcal{C}$ or $\rho = \sigma$ with success probability at least $1-\delta$.
\end{lmma}
This lemma establishes a reduction: lower bounds for decision problems automatically imply lower bounds for $\QSQ$ learning. The next step is to show that the complexity of such decision problems is itself controlled by the $\QSD$.
\begin{lmma}[$\QSD$ lower bounds decision complexity~\cite{arunachalam2023roleentanglementstatisticslearning}]\label{lmma: qsd_bound_decide}
If an algorithm solves the decision problem of whether $\rho \in \mathcal{C}$ or $\rho = \sigma \notin \mathcal{C}$  with success probability at least $1-\delta$, then
\[
   \QQC^{\delta}_{\tau}(\mathcal{C},\sigma) \;\ge\; (1-2\delta)\,\QSD_{\tau}(\mathcal{C},\sigma).
\]
\end{lmma}
Taken together, the lemmas show the complete chain of reductions:
\[
   \QSD \;\;\longrightarrow\;\; \text{Decision Complexity} \;\;\longrightarrow\;\; \text{$\QSQ$ Learning Complexity}.
\]
Thus, by proving that $\QSD_{\tau}(\mathcal{C},\sigma)$ is large for a given concept class, we immediately obtain a lower bound on the $\QSQ$ complexity of learning that class. Combining these two lemmas, we get the proof for Lemma~\ref{lmma:qsd_bound_qsq}.

While $\QSD$ provides a powerful conceptual tool, it is often difficult to compute directly. To make the framework usable in practice, we need bounding techniques that relate $\QSD$ to more tractable combinatorial parameters. Arunachalam \textit{et al.}~\cite{arunachalam2023roleentanglementstatisticslearning} introduced two techniques, namely \emph{variance bound} and \emph{average correlation}, to understand this quantity. They are presented in the following lemmas.

\begin{lmma}[Variance bound~\cite{arunachalam2023roleentanglementstatisticslearning}]\label{lmma: variance_bound_app}
For a concept class $\mathcal{C}$ and a distribution $\mu$ over $\mathcal{C}$ such that $\sigma = \mathbb{E}_{\rho \in \mu}[\rho] \notin \mathcal{C}$, define
\[
   \mathrm{Var}(\mathcal{C}) \;=\;
   \sup_{\|O\|\leq 1} \left( 
      \mathbb{E}_{\rho \sim \mu}\!\left[ \Tr[O\rho]^2 \right] 
      - \left( \mathbb{E}_{\rho \sim \mu} \Tr[O\rho] \right)^2
   \right).
\]
Then, we have:
\[
   \QSD_{\tau}(\mathcal{C},\sigma) \;\ge\; 
   \Omega\!\left(\frac{\tau^2}{\mathrm{Var}(\mathcal{C})}\right).
\]
\end{lmma}
\begin{lmma}[Average-correlation lower bound for $\QSD$ \cite{arunachalam2023roleentanglementstatisticslearning}]
\label{lem:avg-corr-qsd_app}
Let $\mathcal{C}$ be a class of $n$-qubit states and let $\sigma\notin\mathcal{C}$ be full rank. For $\rho\in\mathcal{C}$ define the
$\sigma$-centered, $\sigma$-normalized deviation
\[
\hat\rho \;:=\; \rho\,\sigma^{-1}-\mathbb{I}.
\]
For any finite subfamily $\mathcal{C}'\subseteq\mathcal{C}$, define its \emph{average correlation} w.r.t.\ $\sigma$ as
\[
\gamma(\mathcal{C}',\sigma)
\;:=\;
\frac{1}{|\mathcal{C}'|^2}\sum_{\rho_i,\rho_j\in\mathcal{C}'} \bigl|\Tr\!\bigl[\hat\rho_i\,\hat\rho_j\,\sigma\bigr]\bigr|.
\]
For a base family $\mathcal{C}_0\subseteq\mathcal{C}$ and threshold $\tau>0$, set
\[
\kappa^{\gamma}_{\tau}\text{-frac}(\mathcal{C}_0,\sigma)
\;:=\;
\max_{\mathcal{C}'\subseteq\mathcal{C}_0}
\left\{\frac{|\mathcal{C}'|}{|\mathcal{C}_0|}\;:\;\gamma(\mathcal{C}',\sigma)>\tau\right\},
\qquad
\QAC_{\tau}(\mathcal{C},\sigma)
\;:=\;
\sup_{\mathcal{C}_0\subseteq\mathcal{C}}\bigl(\kappa^{\gamma}_{\tau}\text{-frac}(\mathcal{C}_0,\sigma)\bigr)^{-1}.
\]
Then for every $\tau\in(0,1]$, the quantum statistical dimension satisfies
\[
\QSD_{\tau}(\mathcal{C},\sigma)\;\ge\;\QAC_{\tau^{2}}(\mathcal{C},\sigma).
\]
\end{lmma}
The variance method analyzes the fluctuations of measurement outcomes across an ensemble of states: if, for every observable $O$, the variance of $\Tr[O(\rho-\sigma)]$ over $\rho \in \mathcal{C}$ is small, then no single query can reliably distinguish members of $\mathcal{C}$ from the reference state $\sigma$, forcing the learner to make many queries. The average correlation method, by contrast, considers pairwise overlaps among the states in $\mathcal{C}$: if different $\rho,\rho' \in \mathcal{C}$ yield highly correlated responses to all observables, then even collective querying cannot separate them efficiently, and the $\QSD$ must be large.
\section{Deffered Proofs}
For better readability, we deferred the following proofs to this appendix, restating them for convenience.
\subsection{Proofs of Theorem~\ref{thrm: quantum_parity}}\label{app: quantum_parity}
At a high level, the algorithm consists of two components: an efficient state-preparation procedure that coherently computes 
$c_A \bmod 2$ (Lemma~\ref{lemma: state_prep}), followed by a quantum Fourier sampling step that estimates the Fourier mass of a target function $f$ (Lemma~\ref{lmma:fourier-mass-one-query}). We now formalize these ingredients by restating the two lemmas and providing their proofs.
\begingroup
\renewcommand{\thelmma}{\ref{lemma: state_prep}}
\begin{lmma}
    [State preparation] Let $\ket{g_{S}} = \frac{1}{\sqrt{N}}\sum_{A\in\{0,1\}^{n \times n}} \ket{x_A}\ket{g_{S}(A)}$. There exists a $\text{poly}(n)$-gate quantum circuit $\mathsf{PREP}$ such that
    $$
        \mathsf{PREP} \ket{g_{S}}\ket{0^{\otimes n + n\lfloor\log(n)\rfloor + 1}} = \frac{1}{\sqrt{N}} \sum_{A\in\{0,1\}^{n \times n}}\ket{x_A} \ket{\hat{c}_{A}}\ket{g_{S}(A)}
    $$
    where $N=2^{n^2}$.
\end{lmma}
\endgroup
\begin{proof}
    This preparation procedure involves two steps. First, for each register indexed by $i$, we aim to prepare:
    \begin{equation}
        \ket{x_A}\ket{0^{\otimes \lfloor n\log(n)\rfloor}} \mapsto \ket{x_A} \bigotimes_{i\in [n]} \ket{s_i}. 
        \label{eqn: sum_Aij}
    \end{equation}
    where $s_i := \sum_j A_{ij}$. Here, we employ the addition circuit from \cite{draper2000additionquantumcomputer}. The technique utilizes the quantum Fourier transform to eliminate the need for temporary carry bits in traditional adder circuits. We refer the interested reader to the original paper for the details. In total, this step requires $n$ additional circuits, each includes $ n+2\lfloor\log(n)\rfloor$ controlled-rotation and Hadamard gates.
    
    Then a sequence of controlled-X gates is applied on a new register to store values of $\hat{c}_{A}$ controlled by values of $i$, followed by an uncomputation step of $s_i$:
    \begin{equation}
        \ket{x_A}\bigotimes_{i\in [n]}\ket{s_i}\ket{0^{\otimes n+1}} \mapsto \ket{x_A}\ket{\hat{c}_A} 
    \end{equation}
    In particular, for $k$-th qubit in the last register ($k\neq 0$), $n$ multicontrolled-X gates act on the qubit and are controlled by the values of $s_i=k-1$ for every $i\in [n]$. An example of circuit calculating $\ket{[\hat{c}_{A}]_2}$ is presented in Figure~\ref{fig: multi_controlled}. 
    Thus, we need $n^2$ gates for this operation to obtain $\ket{\hat{c}_A}$. This step finishes with an additional application of the operation generating~\eqref{eqn: sum_Aij} to uncompute $s_i$. 

\begin{figure}[!hbt]
\begin{center}
\begin{quantikz}
\lstick[4]{$\ket{s_0}$}  & \ctrl{9} &  & & &\\
 &  \octrl{8} & &  &  &\\
\vdots \\
 & \octrl{6} & &  & & \\
\vdots \\
\lstick[4]{$\ket{s_{n-1}}$}  &  & & \ctrl{4} & &\\
 &   &  &\octrl{3}&  &\\
\vdots \\
 &  & &\octrl{1}  & & \\
\lstick{$\ket{0}$} & \gate{X} &  \hdots & \gate{X}& &\rstick[]{$\ket{[\hat{c}_{A}]_2}$}
\end{quantikz}
\caption{Circuit implementation of calculating $\ket{[\hat{c}_{A}]_2}$.}
\label{fig: multi_controlled}
\end{center}
\end{figure}
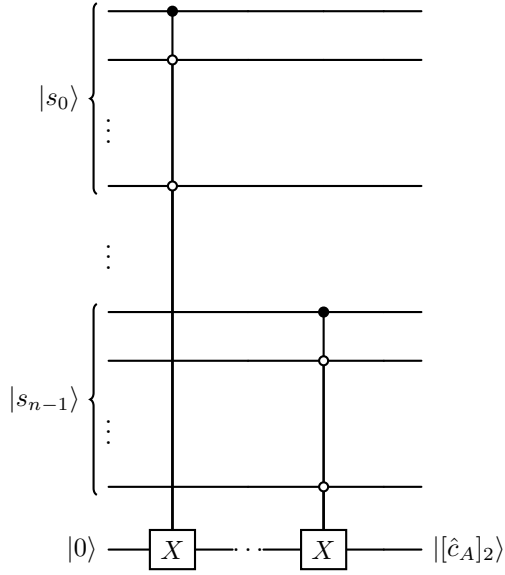
Therefore, our circuit $\mathsf{PREP}$ uses additional $n + n\lfloor\log(n)\rfloor + 1$ qubits and $3n^2 + 4n\lfloor\log(n)\rfloor$ gates to map
    $$
         \ket{g_{S}}\ket{0^{\otimes n + n\lfloor\log(n)\rfloor + 1}} \mapsto \frac{1}{{N}} \sum_{A\in\{0,1\}^{n \times n}} \ket{x_A} \ket{\hat{c}_{A}}\ket{g_{S}(A)}
    $$
\end{proof}

\begingroup
\renewcommand{\thelmma}{\ref{lmma:fourier-mass-one-query}}
\begin{lmma}
Let $f:\{0,1\}^n\to\{0,1\}$ and let $\ket{\psi_f}\;=\;2^{-n/2}\sum_{x\in\{0,1\}^n}\ket{x}\ket{f(x)}$ be the corresponding quantum example state, for any subset $T\subseteq\{0,1\}^n$, there is an observable $M_T$ with $\|M_T\|\le 1$ such that $
\Tr[M_T\ketbra{\psi_f}{\psi_f}] \;=\; \sum_{S\in T}\widehat{f}(S)^2$.
Consequently, a single call to the $\Qstat(M_T,\tau)$ returns a $\tau$-estimate of $\sum_{S\in T}\widehat{f}(S)^2$.
\end{lmma}
\endgroup

\begin{proof}
Let $M:=\sum_{S\in T}\ketbra{S}{S}$ act on the first $n$ qubits and let $\Pi_1:=\ketbra{1}{1}$ on the last qubit.
Define the observable
\[
M_T \;=\; H^{\otimes (n+1)}\,(\mathbb{I}^{\otimes n}\!\otimes\!\Pi_1)\,M\,(\mathbb{I}^{\otimes n}\!\otimes\!\Pi_1)\,H^{\otimes (n+1)}.
\]
Starting from $\ket{\psi_f}=2^{-n/2}\sum_x\ket{x}\ket{f(x)}$, one has
\[
H^{\otimes(n+1)}\ket{\psi_f}
=\frac{1}{2^{n+1/2}} \sum_{x, y}\sum_{b\in \{0,1\}} (-1)^{x\cdot y + b\cdot f(x)}\ket{y}\ket{b},
\]
so conditioned on the last qubit being $1$, the (unnormalized) state on the first $n$ qubits is
$\ket{\psi_f'}=\sum_{Q}\hat f(Q)\ket{Q}$.
Hence
\[
\bra{\psi_f}M_T\ket{\psi_f}
=\bra{\psi_f'}M\ket{\psi_f'}
=\sum_{S\in T}\hat f(S)^2.
\]
Since $\|M_T\|\le 1$, a single $\Qstat(M_T,\tau)$ call returns a $\tau$-estimate of
$\sum_{S\in T}\hat f(S)^2$, as required.
\end{proof}

Now, we can present the proof for Theorem~\ref{thrm: quantum_parity} as follows. We have each function in $\mathcal{C}_n$ admits the representation
\[
    g_{S}(A) = \hat{S}\cdot \hat{c}_{A} \bmod 2,
\]
where $\hat c_A=c_A\bmod 2$ and $\hat S\in\{0,1\}^{n+1}$ is the indicator vector of
$S$.
Thus, learning $g_S$ reduces to identifying the hidden string $\hat S$.

With $\Qstat$ access, we consider the quantum example state of the form
\[
    \ket{g_{S}} 
    = \frac{1}{\sqrt{2^{n^2}}} \sum_{A \in \{0,1\}^{n \times n}} 
      \ket{x_A} \ket{g_{S}(A)}.
\]
By Lemma~\ref{lemma: state_prep}, there exists an efficient state-preparation
procedure mapping $\ket{g_S}$ to
\begin{equation}
    \ket{g'_{S}} 
    = \frac{1}{\sqrt{2^{n^2}}} \sum_{A \in \{0,1\}^{n \times n}} 
      \ket{x_A} \ket{\hat{c}_{A}}\ket{g_{S}(A)},
    \label{eqn: state_before_fourier}
\end{equation}
which serves as input to a Fourier-sampling routine.

We analyze $g_S$ via its Walsh--Hadamard Fourier expansion.
For each coordinate $i$, the Fourier influence
\[
\mathrm{Inf}_i(g_S)=\sum_{\sigma:\,\sigma_i=1}\hat g_S(\sigma)^2
\]
equals $1$ if $i\in\mathrm{supp}(\hat S)$ and $0$ otherwise.
Hence, estimating $\mathrm{Inf}_i(g_S)$ determines the support of $\hat S$.

For each $i\in[n+1]$, let $T_i=\{\sigma:\sigma_i=1\}$.
By Lemma~\ref{lmma: QSQ_obs}, there exists an observable $M_{T_i}$ such that a
single $\Qstat(M_{T_i},\tau)$ query returns a $\tau$-estimate of
$\mathrm{Inf}_i(g_S)$.
Since the expectation gap is $1$ versus $0$, any $\tau<\tfrac12$ suffices to
distinguish the two cases.

The corresponding Fourier-sampling observable is
\[
B_i := (I^{\otimes n^2} \otimes H^{\otimes(n+2)})\cdot
      M_{T_i}\cdot
      (I^{\otimes n^2} \otimes H^{\otimes(n+2)}),
\]
which can be implemented with $\text{poly}(n)$ gates.
Combining $B_i$ with the state-preparation circuit $\mathsf{PREP}$ yields the
observable
\[
O_i
= (\bra{0}^{m}\!\otimes I)\,
  \mathsf{PREP}^{\dagger}\,
  (I\otimes B_i)\,
  \mathsf{PREP}\,
  (\ket{0}^{m}\!\otimes I),
\]
where $m=n+n\lfloor\log n\rfloor+1$.
Each $O_i$ is efficiently implementable and allows estimation of
$\mathrm{Inf}_i(g_S)$ with a single $\QSQ$ query.

Querying all $i\in[n+1]$ recovers $\hat S$ exactly using $\mathcal{O}(n)$
$\Qstat(\tau)$ queries, completing the proof.

\subsection{Pairwise independent of symmetric Boolean functions}\label{app: pairwise}
\begingroup
\renewcommand{\thelmma}{\ref{lmma: H_pairwise}}
\begin{lmma}
 Consider $\Vert p_{\mathcal{O}_\rho}\Vert_2^2 = \sum_{O_k \in \mathcal{O}_{\rho}} (|O_k|/|\mathcal{X}|)^2$, with $(1- \Vert p_{\mathcal{O}_\rho}\Vert_2^2)$ over the choice of $x, x'\in \mathcal{X}$, for every $f\in \text{Unif}(\mathcal{C}_{\rho})$ the distribution of $(f(x), f(x'))$ is equal to uniform distribution over $\{0, 1\}^{\otimes 2}$
\end{lmma}
\endgroup
\begin{proof}
    Here, we follow the proof from \cite{kiani2024hardness}. Consider any function $f\in \text{Unif}(\mathcal{C}_\rho)$, as the function is symmetric, then $f$ is constant on orbits. Thus, for every value of $x\neq x'$ not coming from the same orbit, the distribution of $(f(x), f(x')$ is equal to the uniform distribution over $\{0,1\}^2$. This occurs with probability:
    $$
        \eta = \sum_{O_k \in O_{\rho}} \left(\frac{|O_k|}{|\mathcal{X}|}\right)^2 = \Vert p_{\mathcal{O}_\rho}\Vert_2^2.
    $$
    Thus, $\mathcal{C}_{\rho}$ is an $(1- \Vert p_{\mathcal{O}_\rho}\Vert_2^2)$-pairwise independent function family.
\end{proof}

\subsection{Proof of Theorem~\ref{thrm: qsq_lb_general}}\label{app: proof_general_obs}
The proof of Theorem~\ref{thrm: qsq_lb_general} follows the same overall strategy as
Theorem~\ref{thrm: qsq_lb_diagonal}, but extends the argument to handle
\emph{general (possibly non-diagonal) observables}. The main additional challenge
is to control the contributions arising from non-diagonal terms in the observable
decomposition.

A key ingredient is the pairwise-independence structure induced by symmetry in the
function class $\mathcal{C}_{\rho}$ (see Definition~\ref{def: sym_func}).
Specifically, for a uniformly random function $f\in\mathcal{C}_{\rho}$, the joint
behavior of $f(x)$ and $f(x')$ is completely determined by whether $x$ and $x'$
lie in the same group orbit. This is captured by the following lemma.

\begin{lmma}
    For every $f \in \mathcal{C}_\rho$, with $x, x'$ are drawn independently from $\mathcal{X}$, the quantity
    $$
    \mathbb{E}_f[(-1)^{f(x)+f(x')}] = \left\{ \begin{array}{cc}
         1 & \text{if $x$ and $x'$ belong to the same orbit} \\
         0 & \text{otherwise}
    \end{array}\right.
    $$
    \label{lmma: 2}
\end{lmma}
\begin{proof}
    From Lemma~\ref{lmma: H_pairwise}, we know that the distribution of $(f(x), f(x'))$ will equal the uniform distribution on $\{0, 1\}^2$, whenever $x$ and $x'$ belong from different orbit. In such case, 
    $$\mathbb{E}_{f\in \mathrm{Unif}(\mathcal{C}_\rho)}[(-1)^{f(x)+f(x')}] = \mathbb{E}_{y, y'\sim \{0, 1\}^2}[(-1)^{y+y'}] = 0$$

    On the other hand, if the function $f$ is constant on a single orbit, then:
    $$\mathbb{E}_{f\sim \mathrm{Unif}(\mathcal{C}_\rho)}[(-1)^{f(x)+f(x')}] = \mathbb{E}_{f\in \mathrm{Unif}(\mathcal{C}_\rho)}[(-1)^{2f(x)}] = 1.$$
\end{proof}

As the last ingredient,  we use a standard phase-kickback identity to decompose the quantum example state into a uniform component and a phase-encoded component.
This decomposition enables the analysis of general observables acting on
$\ket{\psi_f}$ in later variance and symmetry arguments.

\begin{claim}
    For every $f: \{0, 1\}^n \to \{0, 1\}$ and let $\ket{\psi_f} = \frac{1}{\sqrt{|\mathcal{X}|}}\sum_{x\in\mathcal{X}} \ket{x}\ket{f(x)}$, $\ket{u} = \frac{1}{\sqrt{|\mathcal{X}|}}\sum_{x\in\mathcal{X}} \ket{x}$, and $\ket{\phi_f} = \frac{1}{\sqrt{|\mathcal{X}|}}\sum_{x\in\mathcal{X}} (-1)^{f(x)}\ket{x}$, then:
    $$
        \ket{\psi_f}\bra{\psi_f} = \frac{1}{2}\left(\ketbra{\phi_f}{\phi_f}\otimes\ketbra{-}{-}+\ketbra{\phi_f}{u}\otimes\ketbra{-}{+}+\ketbra{u}{\phi_f}\otimes\ketbra{+}{-}+\ketbra{u}{u}\otimes\ketbra{+}{+}\right) 
    $$
    \label{claim: phase_kickback}
\end{claim}
\begin{proof}
    We first consider:
    \begin{align}
        (\mathbb{I}\otimes H)\ketbra{\psi_f}{\psi_f}(\mathbb{I}\otimes H) &= \frac{1}{2|\mathcal{X}|}\sum_{x, x'\in \mathcal{X}, a,b\in \{0,1\}} (-1)^{a\cdot f(x) + b\cdot f(x')}\ketbra{x, a}{x', b} \\
        &= \frac{1}{2}\left( \ketbra{\phi_f}{\phi_f}\otimes\ketbra{1}{1} + \ketbra{\phi_f}{u}\otimes \ketbra{1}{0}+\ketbra{u}{\phi_f}\otimes \ketbra{0}{1} + \ketbra{u}{u}\otimes \ketbra{0}{0} \right) \\
    \end{align}
    Thus,
    \begin{equation}
        \ketbra{\psi_f}{\psi_f} = \frac{1}{2}\left(\ketbra{\phi_f}{\phi_f}\otimes\ketbra{-}{-}+\ketbra{\phi_f}{u}\otimes\ketbra{-}{+}+\ketbra{u}{\phi_f}\otimes\ketbra{+}{-}+\ketbra{u}{u}\otimes\ketbra{+}{+}\right) 
    \end{equation}
\end{proof}

We can now proceed to the proof of the main theorem. We first convert the learning problem to a decision problem between $\rho_f = \ket{\psi_f}\bra{\psi_f}$ where $f$ is any state in $\mathcal{C}_\rho$, and an adversarial state $\sigma = \mathbb{E}_{f\sim \mathcal{C}_\rho}[\ket{\psi_f}\bra{\psi_f}]$.

    Using Lemma~\ref{lmma: learn_as_decide} and~\ref{lmma: qsd_bound_decide}, we have:
    $$
    \QSQ^{\varepsilon,\delta}_{\tau}(\mathcal{C}_\rho) \geq (1-2\delta)\QSD_{\tau}(\mathcal{C}_\rho, \sigma) -1
    $$

    To apply this relationship, we need $\min_f \TRD(\rho_f, \sigma) > 2(\tau +\epsilon)$. We observe that:
    \begin{align}
        \TRD(\rho_f, \sigma) \geq 1 - \sqrt{\mathbb{E}_{f'} [|\braket{\psi_f}{\psi_{f'}}|^2]} \geq 1-\sqrt{\frac{1}{2}}
    \end{align}
    
    Fixing $\epsilon = 0.05$ and $\delta = 0.01$, then the Lemma~\ref{lmma: learn_as_decide} holds if $\tau < 0.096$. 
Now, we use the variance bound for an arbitrary observable $O$ as follows:
\begin{equation}
    \mathrm{Var}(\mathcal{C}_\rho) = \mathbb{E}_f[\Tr[O\ket{\psi_f}\bra{\psi_f}]^2] - (\mathbb{E}_f[\Tr[O\ket{\psi_f}\bra{\psi_f}]])^2
    \label{eqn: variance}
\end{equation}
To understand this quantity, we first apply the phase-kickback trick in Claim~\ref{claim: phase_kickback}:
$$
\ket{\psi_f}\bra{\psi_f} = \frac{1}{2}\left(\ket{\phi_f}\bra{\phi_f}\otimes\ket{-}\bra{-}+\ket{\phi_f}\bra{u}\otimes\ket{-}\bra{+}+\ket{u}\bra{\phi_f}\otimes\ket{+}\bra{-}+\ket{u}\bra{u}\otimes\ket{+}\bra{+}\right)
$$
From this identity, we can generally decompose the observable $O$ into the basis of $\{\ket{+}, \ket{-}\}$ on the last qubit: $O = \sum_{a, b} O_{a, b}\ket{a}\bra{b}$, where $a, b \in \{+, -\}$ and $O_{a,b} \in \mathbb{C}^{2^n \times 2^n}$ and $\Vert O_{a,b}\Vert \leq 1$. For the off-diagonal $O_{+, -}$ and $O_{-, +}$, we do not need them to be Hermitian, instead $O_{+, -}=O^{\dagger}_{-, +}$. For convenience, we rewrite $O_1 = O_{+, +}$, $O_2 = O_{+, -}$, $O_3 = O_{-, +}$, and $O_4 = O_{-, -}$    Thus, we have:
\begin{align}
    \Tr[O\ket{\psi_f}\bra{\psi_f}] = \frac{1}{2}\sum_{i\in [4]} \Tr[O_i \rho^{f}_i]
\end{align}
where $\rho^f_1 = \ket{u}\bra{u}$, $\rho^f_2 = \ket{u}\bra{\phi_f}$, $\rho^f_3 = \ket{\phi_f}\bra{u}$, and $\rho^f_4 = \ket{\phi_f}\bra{\phi_f}$. Replace this to~\eqref{eqn: variance}:
\begin{align}
    \mathrm{Var}(\mathcal{C}_\rho) = \frac{1}{4}\left(\mathbb{E}_f\left[\sum_{i,j}\Tr[O_i\rho^{f}_i]\cdot \Tr[O_j\rho^{f}_j]\right] - \left(\sum_{i}\Tr[O_i\mathbb{E}_f[\rho^f_i]]\right)^2\right)
    \label{eqn: variance_general_O}
\end{align}
First, we notice that $\ket{u}\bra{u}$ does not depend on $f$. Therefore, with either $i=1$ or $j=1$, the term equals $0$. We are particularly interested in $i, j\in \{2, 3, 4\}$ as follows
\paragraph{(1) $\boldsymbol{i=j=4}$:} We first consider
\begin{align}
    \mathbb{E}_f[\rho^f_4] = \frac{1}{|\mathcal{X}|}\sum_{x, x'}\mathbb{E}_f[(-1)^{f(x)+f(x')}]\ket{x}\bra{x'}
\end{align}
Using Lemma~\ref{lmma: 2}, we have:
\begin{align}
     \mathbb{E}_f[\rho^f_4] = \frac{1}{|\mathcal{X}|}\sum_{O_k \in \mathcal{O}_\rho} \sum_{x, x'\in O_k} \ket{x}\bra{x'}
\end{align}
Next, observe that
\begin{align}
    \mathbb{E}_f[\rho^f_4\otimes\rho^f_4] = \frac{1}{|\mathcal{X}|^2} \sum_{x, x', v, v' \in \mathcal{X}} \mathbb{E}_f[(-1)^{f(x)+f(x')+f(v)+f(v')}]\ket{x, v}\bra{x, v}.
\end{align}
Here, we see that the expectation $\mathbb{E}_f[(-1)^{f(x)+f(x')+f(v)+f(v')}]$ will vanish if $x, x', v, v'$ belong to four different orbits. Thus, we are particularly interested in the case where at least $2$ variables belong to the same orbit. Without loss of generality, we consider $x, x'$ in the same orbit, then we have:

\begin{align}
    \mathbb{E}_f[\rho^f_4\otimes\rho^f_4] &= \frac{1}{|\mathcal{X}|^2} \sum_{O_k, O_{k'} \in \mathcal{O}_\rho} \sum_{x, x'\in O_k} \sum_{v, v'\in O_{k'}} \mathbb{E}_f[(-1)^{f(v)+f(v')}]\ket{x, v}\bra{x', v'} \\
    &= \frac{1}{|\mathcal{X}|^2} \sum_{O_k, O_{k'} \in \mathcal{O}_\rho} \sum_{x, x'\in O_k} \sum_{v, v'\in \mathcal{O}_{k'}} \ket{x, v}\bra{x', v'}
\end{align}
where the last equality uses Lemma~\ref{lmma: 2}.

Therefore,
\begin{align}
    &\mathbb{E}_f[\Tr[O_4 \otimes O_4 \cdot \rho^f_4 \otimes  \rho^f_4 ]] - \left(\Tr[O_f\cdot \mathbb{E}_f[\rho^f_4]]\right)^2  \\
    &=\frac{1}{|\mathcal{X}|^2} \sum_{O_k, O_{k'} \in \mathcal{O}_\rho} \sum_{x, x'\in O_k} \sum_{v, v'\in \mathcal{O}_{k'}} \bra{x}O_4\ket{x'} \cdot \bra{v}O_4\ket{v'}  - \left(\frac{1}{|\mathcal{X}|}\sum_{O_k \in \mathcal{O}_\rho} \sum_{x, x'\in O_k} \bra{x}O_4\ket{x'}\right)^2 \\ 
    &= 0
\end{align}

\paragraph{(2) $\boldsymbol{i=j=2}$:}
We first consider 
\begin{equation}
    \mathbb{E}_f[\rho^f_2] = \frac{1}{|\mathcal{X}|}\sum_{x, x'\in \mathcal{X}} \mathbb{E}_f[(-1)^{f(x')}]\ket{x}\bra{x'} = 0
\end{equation}
as $\mathbb{E}_f[(-1)^{f(x')}] = 0$ for every $x$.

Next, we consider:
\begin{align}
    \mathbb{E}_f[\rho^f_2\otimes \rho^f_2] &= \frac{1}{|\mathcal{X}|^2} \sum_{x, x', v, v'\in \mathcal{X}} \mathbb{E}_f[(-1)^{f(x')+f(v')}]\ket{x, v}\bra{x', v'} \\
    &= \frac{1}{|\mathcal{X}|^2} \sum_{x, v\in \mathcal{X}} \sum_{O_k\in \mathcal{O}_\rho}\sum_{x', v'\in O_k}\ket{x, v}\bra{x', v'}
\end{align}

Therefore, we have:
\begin{align}
    \mathbb{E}_f[\Tr[O_2 \otimes O_2 \cdot \rho^f_2 \otimes  \rho^f_2 ]] - (\Tr[O_f\cdot \mathbb{E}_f[\rho^f_2]])^2 &= \frac{1}{|\mathcal{X}|^2} \sum_{x, v\in \mathcal{X}} \sum_{O_k\in \mathcal{O}_\rho}\sum_{x', v'\in O_k}\bra{x}O_2\ket{x'}\cdot\bra{v}O_2\ket{v'} 
\end{align}
\paragraph{(3) $\boldsymbol{i=j=3}$:} Same as the case $i=j=2$, we have:
\begin{align}
    \mathbb{E}_f[\Tr[O_3 \otimes O_3 \cdot \rho^f_3 \otimes  \rho^f_3 ]] - (\Tr[O_f\cdot \mathbb{E}_f[\rho^f_3]])^2 &= \frac{1}{|\mathcal{X}|^2} \sum_{x', v'\in \mathcal{X}} \sum_{O_k\in \mathcal{O}_\rho}\sum_{x, v\in O_k}\bra{x}O_3\ket{x'}\cdot\bra{v}O_3\ket{v'} 
\end{align}
\paragraph{(4) $\boldsymbol{i=2, j=3}$ or $\boldsymbol{i=3, j=2}$:}

As both $\mathbb{E}_f[\rho^f_2] = \mathbb{E}_f[\rho^f_3] = 0$, we are only interested in $\mathbb{E}_f[\rho^f_2\otimes \rho^f_3]$ or $\mathbb{E}_f[\rho^f_3\otimes \rho^f_2]$:

\begin{align}
    \mathbb{E}_f[\rho^f_2\otimes \rho^f_3] &= \frac{1}{|\mathcal{X}|^2}\sum_{x, x', v, v'\in \mathcal{X}} \mathbb{E}_f[(-1)^{f(x')+f(v)}]\ket{x, v}\bra{x', v'} \\ 
    &= \frac{1}{|\mathcal{X}|^2}\sum_{x, v'\in \mathcal{X}} \sum_{O_k\in \mathcal{O}_\rho}\sum_{x', v\in O_k}\ket{x, v}\bra{x', v'}
\end{align}
and 
\begin{align}
    \mathbb{E}_f[\rho^f_3\otimes \rho^f_2] &= \frac{1}{|\mathcal{X}|^2}\sum_{x, x', v, v'\in \mathcal{X}} \mathbb{E}_f[(-1)^{f(x)+f(v')}]\ket{x, v}\bra{x', v'} \\ 
    &= \frac{1}{|\mathcal{X}|^2}\sum_{x', v\in \mathcal{X}} \sum_{O_k\in \mathcal{O}_\rho}\sum_{x, v'\in O_k}\ket{x, v}\bra{x', v'}
\end{align}
Therefore
\begin{align}
    \mathbb{E}_f[\Tr[O_2 \otimes O_3 \cdot \rho^f_2 \otimes  \rho^f_3 ]] - (\Tr[O_f\cdot \mathbb{E}_f[\rho^f_2]])^2 = \frac{1}{|\mathcal{X}|^2}\sum_{x, v'\in \mathcal{X}} \sum_{O_k\in \mathcal{O}_\rho}\sum_{x', v\in O_k}\bra{x}O_2\ket{x'}\cdot\bra{v}O_3\ket{v'}
\end{align}
and 
\begin{align}
    \mathbb{E}_f[\Tr[O_3 \otimes O_2 \cdot \rho^f_3 \otimes  \rho^f_2 ]] - (\Tr[O_f\cdot \mathbb{E}_f[\rho^f_3]])^2 = \frac{1}{|\mathcal{X}|^2}\sum_{x', v\in \mathcal{X}} \sum_{O_k\in \mathcal{O}_\rho}\sum_{x, v'\in O_k}\bra{x}O_3\ket{x'}\cdot\bra{v}O_2\ket{v'}
\end{align}
\paragraph{(5) $\boldsymbol{i=4, j=2,3}$: } 
We first consider 
\begin{equation}
    \mathbb{E}_f[\rho^f_4\otimes\rho^f_2] = \frac{1}{|\mathcal{X}|^2}\sum_{x, x', v, v'} \mathbb{E}_f[(-1)^{f(x)+f(x')+f(v')}]\ket{x, v}\bra{x, v'} = 0
\end{equation}
The second equality holds because $\mathbb{E}_f[(-1)^{f(x)+f(x')+f(v')}] = 0$ based on the observation that, if there are at least 2 variables that belong to the same orbit, the expectation is reduced to $\mathbb{E}_f[(-1)^{f(x)}]$ and equals $0$. Otherwise, all three variables belong to different orbits, using the same reasoning as Lemma~\ref{lmma: 2}, the distribution of $(f(x), f(x'), f(v'))$ will be uniform distribution over $\{0,1\}^{\otimes 3}$ and $\mathbb{E}_f[(-1)^{f(x)+f(x')+f(v')}] = 0$.

Similarly, we have:
\begin{equation}
    \mathbb{E}_f[\rho^f_4\otimes\rho^f_3] = 0
\end{equation}

Putting all this together in~\eqref{eqn: variance_general_O}, we have:
\begin{align}
    \mathrm{Var}(\mathcal{C}_\rho) &= \frac{1}{4|\mathcal{X}|} \sum_{O_k \in \mathcal{O}_{\rho}} \sum_{x, v\in O_k}\bra{x}O_3\ket{u}\cdot\bra{v}O_3\ket{u} + \bra{u}O_2\ket{x}\cdot\bra{u}O_2\ket{v} \\
    &\quad + 
    \bra{x}O_3\ket{u}\cdot\bra{u}O_2\ket{v} +
    \bra{u}O_2\ket{x}\cdot\bra{v}O_3\ket{u} \\
    &= \frac{1}{4|\mathcal{X}|}\sum_{O_k \in \mathcal{O}_\rho} \left( \sum_{x\in O_k} \bra{x}O_3\ket{u} + \bra{u}O_2\ket{x}\right)^2 \\
    &= \frac{1}{|\mathcal{X}|}\sum_{O_k \in \mathcal{O}_\rho} \left( \sum_{x\in O_k}  \mathrm{Re}(\bra{u}O_2\ket{x})\right)^2 \\
    &\leq \frac{1}{|\mathcal{X}|}\sum_{O_k \in \mathcal{O}_\rho} \left( \sum_{x\in O_k}  |\bra{u}O_2\ket{x}|\right)^2 \label{eqn: temp1}
\end{align}
Here, we apply some facts that $O_3 = O^{\dagger}_2$ and $\mathrm{Re}(z) \leq |z|$ for any $z\in \mathbb{C}$. Next, applying the Cauchy-Schwarz inequality to~\eqref{eqn: temp1}, we have
\begin{align}
    \mathrm{Var}(\mathcal{C}_\rho) &\leq \frac{1}{|\mathcal{X}|}\sum_{O_k \in \mathcal{O}_\rho} |O_k|  \sum_{x\in O_k}  |\bra{u}O_2\ket{x}|^2 \\
    &\leq \frac{\max_k{|O_k|}}{|\mathcal{X}|}\sum_{x\in\mathcal{X}}|\bra{u}O_2\ket{x}|^2 \\
    &\leq \frac{\max_k{|O_k|}}{|\mathcal{X}|}\Vert O_2\ket{u}\Vert^2_2 \\
    &\leq\frac{\max_k{|O_k|}}{|\mathcal{X}|}\Vert O_2\Vert^2 \\
    &\leq \frac{\max_k{|O_k|}}{|\mathcal{X}|}
\end{align}
Now we tighten the inequality by choosing an explicit construction of $O_2$. Let $O^\star \in \mathcal{O}_\rho$ be an orbit of maximum size,
\[
|O^\star| = \max_{O_k \in \mathcal{O}_\rho} |O_k|.
\]
Define the normalized uniform superposition over this orbit by
\[
\ket{w}
\;:=\;
\frac{1}{\sqrt{|O^\star|}}
\sum_{x \in O^\star} \ket{x}.
\]
Let $\ket{u}$ denote the fixed reference basis state used in the block decomposition of the observable.
We define
\[
O_2 \;:=\; \ket{u}\!\bra{w},
\qquad
O_3 \;:=\; O_2^\dagger = \ket{w}\!\bra{u}.
\]
Then $\|O_2\| = 1$, and for all $x \in \mathcal{X}$,
\[
\bra{u} O_2 \ket{x}
=
\bra{w}x\rangle
=
\begin{cases}
\displaystyle \frac{1}{\sqrt{|O^\star|}}, & x \in O^\star, \\[6pt]
0, & \text{otherwise}.
\end{cases}
\]
Substituting this into the variance expression yields
\[
\mathrm{Var}(\mathcal{C}_\rho)
=
\frac{1}{|\mathcal{X}|}
\left(
\sum_{x \in O^\star} \frac{1}{\sqrt{|O^\star|}}
\right)^2
=
\frac{|O^\star|}{|\mathcal{X}|},
\]
Thus, the variance bound is tight and $\mathrm{Var}(\mathcal{C}_\rho) = \Theta(\max_k |O_k|/|\mathcal{X}|)$.

Finally, using Lemma~\ref{lmma: learn_as_decide},~\ref{lmma: qsd_bound_decide},~\ref{lmma: variance_bound}, we finish the proof.

\subsection{Proof of Theorem~\ref{thm:tolerance_separation}} \label{app: tolerance_separation}
Our proof is based on the reduction to distinguishing two pure states. In this setting, a common figure of merit is the trace distance $\TRD(\cdot,\cdot)$, which is attained by the \emph{Helstrom measurement}. Concretely, letting $X:=\ket{\psi_f}\!\bra{\psi_f}-\ket{\psi_g}\!\bra{\psi_g}$, the optimal observable is the sign of $X$, i.e., the difference of the projectors onto the positive and negative eigenspaces of $X$ (Jordan–Hahn decomposition). Because our $\QSQ$ oracle accepts arbitrary bounded observables, querying with this optimal observable realizes the information-theoretic limit for distinguishing $\ket{\psi_f}$ from $\ket{\psi_g}$.  We formalize this below.
\begin{lmma}
    Let $\ket{\psi_f}$ and $\ket{\psi_g}$ be two quantum states, there always exists an optimal measurement $O^*$ such that:
    $$
        |\Tr[O^* (\ketbra{\psi_f}{\psi_f}-\ketbra{\psi_g}{\psi_g})]| = 2d_{\text{tr}}(\ket{\psi_f}, \ket{\psi_g})
    $$
    \label{lmma: optimal_measurement}
\end{lmma}
\begin{proof}
    The construction of $O^*$ relies on the Jordan–Hahn decomposition, which we aim to find the projection to the positive and negative parts of $X = \ket{\psi_f}\bra{\psi_f}-\ket{\psi_g}\bra{\psi_g}$.

    Here we define:
    $$
    \ket{\phi} = \frac{\ket{\psi_f} + \ket{\psi_g}}{\sqrt{2+2\braket{\psi_f}{\psi_g}}} \quad \text{and}\quad \ket{\phi^{\perp}} = \frac{\ket{\psi_f} - \ket{\psi_g}}{\sqrt{2-2\braket{\psi_f}{\psi_g}}}
    $$
    It is easy to check that $\braket{\phi}{\phi^\perp} = 0$. Denote $\ket{P} = \frac{\ket{\phi}+ \ket{\phi^\perp}}{\sqrt{2}}$ and $\ket{P^\perp} = \frac{\ket{\phi}- \ket{\phi^\perp}}{\sqrt{2}}$ we have:
    $$
        X\ket{P} = d_{\text{tr}}(\ket{\psi_f}, \ket{\psi_g})\ket{P} \quad \text{and} \quad  X\ket{P^\perp} = -d_{\text{tr}}(\ket{\psi_f}, \ket{\psi_g})\ket{P\perp}
    $$
    Thus, we observe that if $O^* := \ket{P}\bra{P} - \ket{P^\perp}\bra{P^\perp}$, we have:
    $$
    \Tr[O^*X] = 2d_{\text{tr}}(\ket{\psi_f}, \ket{\psi_g}).
    $$
    We finish the proof.
    
\end{proof}
Thus, the proof for Theorem~\ref{thm:tolerance_separation} proceeds by establishing an upper bound on the distinguishing power of any classical $\SQ$ query and a lower bound on the distinguishing power of specific Helstrom measurements.

\paragraph{Classical $\SQ$ Hardness.}
Consider any two functions $f, g \in \mathcal{S} \cup \{0\}$ (where $0$ denotes the zero function). For a classical statistical query defined by a bounded function $\phi: \mathcal{X} \times \{0,1\} \to [-1,1]$, the difference in expected values is:
\begin{equation}
\Delta_{SQ}(f, g) = \left| \mathbb{E}_{x \sim \mathcal{D}}[\phi(x, f(x))] - \mathbb{E}_{x \sim \mathcal{D}}[\phi(x, g(x))] \right|
\end{equation}
Since $f$ and $g$ differ only on the union of their supports $S_f \cup S_g$, and are identical (both zero) elsewhere, the expectation difference is non-zero only on this region.
\begin{equation}
\Delta_{SQ}(f, g) = \left| \frac{1}{|\mathcal{X}|} \sum_{x \in S_f \cup S_g} (\phi(x, f(x)) - \phi(x, g(x))) \right|
\end{equation}
The maximum possible value of the term $(\phi(x, y) - \phi(x, y'))$ is 2 (since $\phi \in [-1,1]$). Thus, we can bound the statistical distance by the measure of the support. Comparing any function $f \in \mathcal{S}$ to the zero function $g=0$:
\begin{equation}
\Delta_{SQ}(f, 0) \le \frac{|S_f|}{|\mathcal{X}|} \cdot 2 = 2\zeta.
\end{equation}
If the query tolerance $\tau > 2\zeta$, the oracle can validly return the value corresponding to the zero function for every query, regardless of whether the true target is $f$ or $0$. Specifically, the oracle can define a ``null'' response $\alpha_0 = \mathbb{E}_{x}[\phi(x, 0)]$. Since $|\alpha_0 - \mathbb{E}_{x}[\phi(x, f(x))]| \le 2\zeta < \tau$, $\alpha_0$ is a valid response for any $f \in \mathcal{S}$. Consequently, the learner effectively observes responses consistent with the zero function for all targets, making learning impossible.

\paragraph{$\QSQ$ Learnability.}
We now construct a quantum learner. The quantum example state for a function $f$ is $|\psi_f\rangle = \frac{1}{\sqrt{|\mathcal{X}|}} \sum_x |x\rangle |f(x)\rangle$.
To distinguish $f$ from the zero function state $|\psi_0\rangle$, we compute their inner product. Since $f(x) \in \{0,1\}$, the states are orthogonal on inputs $x$ where $f(x)=1$, and identical where $f(x)=0$:
\begin{equation}
\langle \psi_f | \psi_0 \rangle = \frac{1}{|\mathcal{X}|} \sum_{x \in \mathcal{X}} \langle f(x) | 0 \rangle = \frac{1}{|\mathcal{X}|} \sum_{x \notin S_f} 1 = 1 - \zeta.
\end{equation}
The trace distance $d_{tr}$, which governs the optimal distinguishability via Helstrom measurements, is given by:
\begin{equation}
d_{tr}(|\psi_f\rangle, |\psi_0\rangle) = \sqrt{1 - |\langle \psi_f | \psi_0 \rangle|^2} = \sqrt{1 - (1-\zeta)^2} = \sqrt{2\zeta - \zeta^2}.
\end{equation}
By Lemma~\ref{lmma: optimal_measurement}, there exists an observable $O^*$ such that the gap in expectation values between $|\psi_f\rangle$ and $|\psi_0\rangle$ is exactly $2 d_{tr} = 2\sqrt{2\zeta - \zeta^2}$. However, to strictly separate $f$ from $0$ given a tolerance $\tau$, we require that the tolerance be smaller than the signal strength. In the $\QSQ$ model, successful identification is possible if the trace distance exceeds $\tau$ (allowing the oracle to return distinct ranges).
Thus, provided $\tau < \sqrt{2\zeta - \zeta^2}$, a $\QSQ$ algorithm can distinguish $f$ from $0$.

\paragraph{The Separation Window.}
A quantum advantage exists if the condition for classical failure ($ \tau > 2\zeta$) and quantum success ($\tau < \sqrt{2\zeta - \zeta^2}$) can be simultaneously satisfied. This requires:
\begin{equation}
2\zeta < \sqrt{2\zeta - \zeta^2}
\end{equation}
Squaring both sides (valid since $\zeta > 0$):
\begin{equation}
4\zeta^2 < 2\zeta - \zeta^2 \implies 5\zeta^2 < 2\zeta \implies \zeta < \frac{2}{5} = 0.4.
\end{equation}
Thus, for any orbit fraction $\zeta < 0.4$, there exists a valid gap. For example, if we choose functions supported on orbits of size $\zeta = 0.2$, $\SQ$ fails for any $\tau > 0.4$, while $\QSQ$ succeeds up to $\tau \approx 0.6$. The quantum learner can proceed by a tournament method: for any candidate $f \in \mathcal{S}$, it applies the optimal observable to reject the hypothesis $f$ or $0$. Since pairwise trace distances between disjoint support functions scale similarly, the class is efficiently learnable.
\end{document}